\documentclass[10pt,twocolumn,journal]{IEEEtran}

\usepackage{amsmath,amssymb,amsfonts,amsthm,mathtools,bm}
\usepackage{braket}
\usepackage{algorithm}
\usepackage{algorithmic}
\usepackage{array}
\usepackage{multirow}
\usepackage{booktabs,tabularx,makecell}
\usepackage[caption=false,font=footnotesize]{subfig}
\usepackage[inline,shortlabels]{enumitem}
\usepackage{textcomp}
\usepackage{stfloats}
\usepackage{url}
\usepackage{verbatim}
\usepackage{graphicx}
\usepackage{cite}
\usepackage{microtype}
\usepackage{xcolor}
\usepackage[colorlinks=true,allcolors=blue]{hyperref}
\usepackage{orcidlink}

\theoremstyle{remark}
\newtheorem{remark}{Remark}
\usepackage{tikz}
\usepackage{pgfplots}
\pgfplotsset{compat=1.17}
\usetikzlibrary{arrows,shapes,positioning,calc,patterns,decorations.pathreplacing}
\usetikzlibrary{fit,backgrounds}

\newtheorem{theorem}{Theorem}

\newtheorem{proposition}{Proposition}
\theoremstyle{definition}
\newtheorem{definition}{Definition}

\theoremstyle{remark}

\DeclareMathOperator*{\argmin}{arg\,min}

\title{Distributed Learning of Quantum State \\ Tomography Robust to Readout Errors}

\author{%
Amirhossein Taherpour~\orcidlink{0000-0003-4647-102X},%
\ Alireza Sadeghi~\orcidlink{0000-0003-1280-7592},%
\ and Georgios B.~Giannakis~\orcidlink{0000-0002-0196-0260},~\IEEEmembership{Fellow,~IEEE}%
\thanks{Amirhossein Taherpour, Alireza Sadeghi, and Georgios B.~Giannakis are with the Dept. of ECE, Univ. of Minnesota, Minneapolis, MN 55455, USA (e-mails: taher054@umn.edu; sadeg012@umn.edu; georgios@umn.edu).\\ \textit{Corresponding author: G. B.~Giannakis.}}%
}
\begin{document}

\maketitle

\begin{abstract}
Scalable estimation of quantum states with readout errors is a central challenge in large multiqubit systems. Existing overlapping-tomography methods improve scalability by working with local subsystems, but they usually assume known or separately calibrated measurements. At the same time, readout-estimation methods model measurement errors without enforcing consistency among overlapping regional states. In this context, the present paper introduces a unified framework for joint regional quantum state tomography with readout errors. A multiqubit system is partitioned in overlapping regions, each region is assigned to a local density operator and a local confusion matrix, and neighboring regions are coupled through reduced-state consistency on shared subsystems. This leads to a structured bilinear optimization problem. To solve it, a distributed alternating method is developed in which the state-update step is handled by the alternating direction method of multipliers (ADMM), while the confusion-matrix updates are carried out locally in parallel. Analytical guarantees are also established, including a sufficient condition for local identifiability, local quadratic growth of the population misfit, and convergence of the inner state-update procedure. Simulations on \textsc{Ring}, \textsc{Ladder}, \textsc{Torus}, and \textsc{Hub} graph geometries show that joint estimation improves state recovery over fixed-readout reconstruction, recovers a substantial portion of oracle performance, and reveals a clear tradeoff between state estimation performance, communication, and computation.
\end{abstract}

\begin{IEEEkeywords}
Quantum state tomography, readout error estimation, overlapping tomography, regional quantum estimation, confusion matrix, distributed optimization, alternating direction method of multipliers.
\end{IEEEkeywords}


\section{Introduction}

Accurate characterization of quantum states is central to quantum information processing, both for validating devices  quantum systems and for assessing control, measurement, and computation protocols \cite{NielsenChuang,Watrous}. When the goal is full-state reconstruction, quantum state tomography provides the natural framework, but its statistical and computational cost grows quickly with system size \cite{ParisRehacek2004QuantumStateEstimation,Renes2004SICPOVM}. This has motivated substantial work on sample-efficient tomography, local reconstruction, and overlapping strategies that avoid monolithic global recovery \cite{Cotler2020,Huang2020ClassicalShadows,Araujo2022LocalQuantumOverlappingTomography,Hu2024,Hansenne2025OptimalOverlappingTomography}.

From a signal processing (SP) viewpoint, tomography is a large-scale inverse problem with physical constraints and rapidly growing dimension. This viewpoint appears in recent work on detector tomography, efficient state reconstruction, and SP formulations of quantum estimation \cite{Wang2021TwoStageDetectorTomography,ShanmugamKalyani2023UnrollingSVT,Xue2025GeneralizedQST}. It is also connected to broader developments in low-rank positive-semidefinite recovery, nonsmooth low-rank optimization, and inverse problems with denoising- or learning-based priors \cite{LiSunChi2017LowRankPSDRecovery,TongMaChi2021ScaledSubgradient,VenkatakrishnanBoumanWohlberg2013PnP,ReehorstSchniter2019RED,XueZhengDai2022PhaseRetrievalPriors}. Related SP efforts have also emphasized distributed ADMM-type methods for stitching together local information in structured inverse problems \cite{VourasMishraArtusioGlimpse2023Overview,MurtadaHuRamaRaoSchroeder2023DistributedRadarImaging,kekatos2012distributed}.

In practical quantum systems, the difficulty is not only the size of the state space. Measurement outcomes are also corrupted by readout error, so the observed distribution can differ substantially from the ideal Born distribution. This has led to several approaches to readout-error mitigation, assignment or confusion matrices, detector tomography, unfolding methods, and structured models for correlated readout error \cite{Maciejewski2020,Nachman2020UnfoldingReadoutNoise,Bravyi2021MitigatingMeasurementErrors,Nation2021,Aasen2024,Aasen2025}. At the same time, most of existing works either assume separate calibration data or treats the readout model as known.

A related line of work instead estimates the state and measurement model jointly. This includes detector tomography, self-consistent measurement tomography, gate-set tomography, and simultaneous state--measurement estimation \cite{Wang2021TwoStageDetectorTomography,Nielsen2021GST,Cattaneo2023SelfConsistentMeasurementTomography,Jayakumar2024SimultaneousTomography}. These formulations capture the bilinear coupling between the state and the readout process, but they remain largely global and thus inherit the same scalability limits at large system size.

A different response to the scaling problem is regional or overlapping tomography, where one reconstructs reduced states on smaller subsystems and enforces consistency across overlaps \cite{Cotler2020,Araujo2022LocalQuantumOverlappingTomography,Hu2024,Hansenne2025OptimalOverlappingTomography}. This shifts the exponential dependence from the full system size to the region and overlap sizes. However, this literature usually assumes known or separately calibrated measurements, while the readout literature generally does not enforce consistency (consensus) across overlapping regional states. This leaves a gap between scalable regional tomography and joint state--readout estimation.

In this paper, we develop a framework for \emph{joint regional quantum state and readout estimation}. We partition a multiqubit system into overlapping regions, assign each region a local density operator and a local confusion matrix, and couple neighboring regions through reduced-state consistency (consensus) on shared subsystems. This boils down to a structured bilinear optimization problem that combines regional tomography, readout estimation, and distributed optimization.

\noindent \emph{Contributions.}
Our main contributions are as follows.
\begin{enumerate}[leftmargin=*,label=(\roman*)]
    \item We formulate a multi-region joint estimation problem in which the unknowns are the regional density operators and the regional confusion matrices. The formulation enforces physicality of each regional state, column-stochasticity of each regional confusion matrix, and consensus of reduced states across overlapping regions.

    \item We establish theoretical guarantees for the proposed model and algorithm. In particular, we derive a local identifiability condition for the joint regional formulation, prove local quadratic growth of the population misfit, and establish convergence of the inner state-update step based on overlap-consensus ADMM.

    \item We develop a distributed proximal alternating ADMM method and study its behavior across several graph topologies. The analysis and simulations show how the exponential dependence shifts from the full system size to the region and overlap sizes, when joint state and readout estimation improves performance, and how the regional topology affects the tradeoff between performance, communication, and computation.
\end{enumerate}

\noindent{{\it Notation.}}
Unless otherwise stated, (bold) lowercase letters denote (vectors) scalars, and uppercase bold letters are matrices. The only exception is $\boldsymbol{\rho}$, which is reserved for the quantum state matrix. Calligraphic uppercase letters are for sets, index sets, feasible sets, and Hilbert spaces; the \(i\)-th entry of vector $\mathbf{x}$ is represented as 
\([\mathbf{x}]_i\), and for a matrix \(\mathbf{X}\), its \((i,j)\)-th entry is denoted by \([\mathbf{X}]_{ij}\); matrices 
\(\mathbf{X}^{\top}\) and \(\mathbf{X}^{\dagger}\) represent transpose and conjugate transpose, respectively; \(\operatorname{Tr}(\mathbf{X})\) is the trace; and \(\mathbf{X}\succeq {\bf 0}\) means that \(\mathbf{X}\) is Hermitian positive semidefinite. The symbols \(\|\cdot\|_2\) and \(\|\cdot\|_F\) denote the Euclidean and Frobenius norms, respectively; \(\langle \mathbf{X},\mathbf{Y}\rangle := \operatorname{Tr}(\mathbf{X}^{\dagger}\mathbf{Y})\) is the Frobenius inner product, \(\mathbb{E}[\cdot]\) stands for expectation, and \(\mathbb{I}\{\cdot\}\) is the indicator function, while \(\otimes\) indicates the tensor product, and finally \(\cong\) denotes the isomorphism. 

The rest of the paper is organized as follows.
Section~\ref{sec:basics_background} reviews the basics of quantum and measurement models. Section~\ref{sec:multi_region} presents the multi-region formulation and the distributed (D-) ADMM algorithm. Section~\ref{sec:theory} develops the main theoretical results. Section~\ref{sec:simulations} reports the simulation tests. Finally, Section~\ref{sec:conclusion} concludes the paper.

\section{Preliminaries}
\label{sec:basics_background}

A qubit is the basic unit of quantum information. Unlike a classical bit, which takes only $0$  or $1$ values, a qubit is a unit vector in a two dimensional complex Hilbert space \(
\mathcal{H} :=\mathbb{C}^2
\). The two computational basis vectors of this space are
\[
\ket{0}:=\begin{bmatrix}1\\0\end{bmatrix},
\qquad
\ket{1}:=\begin{bmatrix}0\\1\end{bmatrix}.
\]
Any single-qubit \emph{pure} state $\bm{\psi}$ can be expressed as a superposition of the basis states, which in Dirac notation is written as 
\[
\ket{\bm{\psi}}=\alpha_0\ket{0}+\alpha_1\ket{1},
\;
\alpha_0,\alpha_1\in\mathbb{C}, 
\textrm{~s.t.~}
|\alpha_0|^2+|\alpha_1|^2=1  
\]
where  \(\ket{\bm{\psi}}\) is called \emph{ket} of $\bm{\psi}$, and denotes a column vector in \(\mathcal{H}\), and the corresponding \emph{bra} is its conjugate transpose,
\(
\bra{\bm{\psi}}:=\ket{\bm{\psi}}^\dagger.
\)
For a system of \(q\) qubits, define the Hilbert space
\[
\mathcal{H}_q:=(\mathcal{H})^{\otimes q}\cong\mathbb{C}^{Q},
\textrm{~~where~~}
Q := 2^q.
\] For each binary string
\(
\mathbf{x} :=[x_1,\ldots,x_q]^\top\in\mathcal{X},
\)
where 
\[
\mathcal{X} := \{0,1\}^q 
= \left\{ \mathbf{x} : x_i \in \{0,1\},\; i=1,\dots,q \right\},
\]
the associated computational-basis ket is  
\[
\ket{\mathbf{x}}
:=
\ket{x_1}\otimes\cdots\otimes\ket{x_q}.
\]
Accordingly, the computational basis of \(\mathcal{H}_q\) is
\(
\{\ket{\mathbf{x}}:\mathbf{x}\in\mathcal{X}\},
\)
and any \emph{pure} \(q\)-qubit state admits the expansion
\[
\ket{\bm{\psi}}=\sum_{\mathbf{x}\in\mathcal{X}}\alpha_{\mathbf{x}}\ket{\mathbf{x}},
\alpha_{\mathbf{x}}\in\mathbb{C}, 
\textrm{~~s.t.~}
\sum_{\mathbf{x}\in\mathcal{X}}|\alpha_{\mathbf{x}}|^2=1,
\]
where
\(
\|\ket{\bm{\psi}}\|_2=1
\)
accounts for the unit energy of the state. By contrast, a \emph{mixed} state cannot in general be represented by a single state vector. Instead, it describes a probabilistic ensemble of pure states and is represented by a density operator
\begin{equation}
\label{eq:q_state}
\bm{\rho} \, = \, \sum_j |\alpha_j|^2 \ket{\bm{\psi}_j}\bra{\bm{\psi}_j}
\end{equation}
where
\(
\alpha_j\in\mathbb{C},
\)
\(
\sum_j |\alpha_j|^2=1,
\)
and
\(
\ket{\bm{\psi}_j}\textrm{'s}\in\mathcal{H}_q
\)
are normalized. The density operator $\bm{\rho} \in \mathbb{C}^{q \times q}$ is a complex matrix, sometimes called the mixed state (or quantum state). Using the computational basis, \eqref{eq:q_state} can be rewritten as
\begin{equation}
\bm{\rho} 
\, 
= \!\!\!
\sum_{\mathbf{x},\mathbf{x}'\in\mathcal{X}}
\rho_{\mathbf{x},\mathbf{x}'}\,\ket{\mathbf{x}}\bra{\mathbf{x}'},
\textrm{~where~~}
\rho_{\mathbf{x},\mathbf{x}'}
:=
\bra{\mathbf{x}}\bm{\rho}\ket{\mathbf{x}'}.
\nonumber
\end{equation}
If
\(
\ket{\bm{\psi}_j}
=
\sum_{\boldsymbol{\xi}\in\mathcal{X}}
\alpha_{j,\boldsymbol{\xi}}\ket{\boldsymbol{\xi}},
\) using
\(
\bra{\mathbf{x}}\ket{\boldsymbol{\xi}}=\delta_{\mathbf{x},\boldsymbol{\xi}},
\)
one obtains
\begin{equation}
\label{eq:density_entries_closed_form}
[\bm{\rho} ]_{\mathbf{x}, \mathbf{x}'} := \rho_{\mathbf{x},\mathbf{x}'}
= \sum_j |\alpha_j|^2\,\alpha_{j,\mathbf{x}}\,{\alpha_{j,\mathbf{x}'}}^\dagger.
\end{equation}
\begin{remark}
The diagonal entry \(\rho_{\mathbf{x},\mathbf{x}}\) gives the population of basis state \(\ket{\mathbf{x}}\); that is, the probability of observing \(\mathbf{x}\) in a computational-basis measurement, whereas an off-diagonal entry \(\rho_{\mathbf{x},\mathbf{x}'}\) with \(\mathbf{x}\neq\mathbf{x}'\) represents the \emph{coherence} between states \(\ket{\mathbf{x}}\) and \(\ket{\mathbf{x}'}\).    
\end{remark}
To summarize, if
\(
\{\mathbf{x}_1,\dots,\mathbf{x}_{Q}\}
\)
is a fixed ordering of \(\mathcal{X}\), then \(\bm{\rho}\) has matrix representation
\[
\bm{\rho}
=
\begin{bmatrix}
\rho_{\mathbf{x}_1,\mathbf{x}_1} & \cdots & \rho_{\mathbf{x}_1,\mathbf{x}_{Q}}\\
\vdots & \ddots & \vdots\\
\rho_{\mathbf{x}_{Q},\mathbf{x}_1} & \cdots & \rho_{\mathbf{x}_{Q},\mathbf{x}_{Q}}
\end{bmatrix} \in \mathbb{C}^{Q \times Q},
\]
where \(\rho_{\mathbf{x},\mathbf{x}'}\) is the entry in the row associated with state \(\ket{\mathbf{x}}\) and the column linked to the state \(\ket{\mathbf{x}'}\). Pure states are precisely those density operators of the form
\(
\bm{\rho}=\ket{\bm{\psi}}\bra{\bm{\psi}}
\)\cite{NielsenChuang,Watrous}.

In finite dimensions, the density operator \(\bm{\rho}\) completely characterizes the state of a quantum system. Hence, every state-dependent quantity is determined by \(\bm{\rho}\).  In general, the state of a quantum system at any time \(t\) is represented by a density operator
\(
\bm{\rho}(t)
\)
acting on
\(
\mathcal{H}_q
\).
For a closed system, the evolution from an initial time \(t_0\) to \(t\) can be described using a unitary operator, i.e., there exists
\(
\mathbf{U}(t,t_0)\in\mathbb{C}^{Q\times Q}
\)
satisfying
\(
\mathbf{U}^\dagger(t,t_0)\mathbf{U}(t,t_0)=\mathbf{I},
\)
where
\begin{equation}
\label{eq:density_evolution}
\bm{\rho}(t)=\mathbf{U}(t,t_0)\,\bm{\rho}(t_0)\,\mathbf{U}^\dagger(t,t_0).
\end{equation}
The propagator
\(
\mathbf{U}(t,t_0)
\)
is uniquely determined by the system Hamiltonian
\(
\mathbf{H}(t)\in\mathbb{C}^{Q\times Q},
\)
satisfying
\(
\mathbf{H}(t)=\mathbf{H}^\dagger(t).
\)
For time-independent Hamiltonians \(
\mathbf{H}(t)=\mathbf{H},
\) it has closed form
\begin{equation}
\label{eq:unitary_time_independent}
\mathbf{U}(t,t_0)=\exp\!\left(-\frac{i}{\hbar}(t-t_0)\mathbf{H}\right)
\end{equation}
where \(i :=\sqrt{-1}\), and \(\hbar\) is the reduced Planck constant. Thus, once the initial state
\(
\bm{\rho}(t_0)
\)
and the Hamiltonian
\(
\mathbf{H}
\) are specified, the state
is uniquely determined for every later time \(t > t_0\)  \cite{NielsenChuang}.

In finite dimensions, \(\bm{\rho}\) fully characterizes the quantum state. Hence every state-dependent quantity is determined by \(\bm{\rho}\). To make this explicit, let
\(
\{\mathbf{G}_k\}_{k=1}^{Q^2}
\)
be a Hermitian operator basis on \(\mathcal{H}_q\) that is orthonormal with respect to the Hilbert--Schmidt inner product. Then every density operator admits the expansion
\begin{equation}
\label{eq:rho_operator_expansion}
\bm{\rho}=\sum_{k=1}^{Q^2}\tau_k(\bm{\rho})\,\mathbf{G}_k,
\qquad
\tau_k(\bm{\rho}):=\operatorname{Tr}(\mathbf{G}_k\bm{\rho})
\end{equation}
confirming that 
\(
\{\tau_k(\bm{\rho})\}_{k=1}^{Q^2}
\)
determine \(\bm{\rho}\) uniquely. Any state-dependent quantity can thus be represented as 
\begin{equation}
\label{eq:state_quantity_coordinates}
f(\bm{\rho})=\Phi\bigl(\tau_1(\bm{\rho}),\dots,\tau_{Q^2}(\bm{\rho})\bigr)
\end{equation}
for a suitable function \(\Phi\). For example, if
\(
\mathbf{O}=\sum_{k=1}^{Q^2} o_k\mathbf{G}_k
\)
is a fixed observable, then its expectation is
\(
\langle \mathbf{O}\rangle_{\bm{\rho}}=\operatorname{Tr}(\mathbf{O}\bm{\rho})=\sum_{k=1}^{Q^2}o_k\,\tau_k(\bm{\rho}),
\)
so
\(
\Phi_{\mathbf{O}}(z_1,\dots,z_{Q^2})=\sum_{k=1}^{Q^2}o_k z_k;
\)
likewise, the purity is
\(
\operatorname{Tr}(\bm{\rho}^2)=\sum_{k=1}^{Q^2}\tau_k(\bm{\rho})^2,
\)
so
\(
\Phi_{\mathrm{pur}}(z_1,\dots,z_{Q^2})=\sum_{k=1}^{Q^2} z_k^2.
\)
More generally, measurement outcome probabilities, computational-basis populations, von Neumann entropy, and fidelity with a fixed target state, can all be expressed through the coordinates
\(
\{\tau_k(\bm{\rho})\}_{k=1}^{Q^2}
\)\cite{Watrous}.

\subsection*{Quantum State Estimation with Readout Errors}
\label{sec:state_est}
The density operator \(\bm{\rho}\) is not directly observable. Rather, information about the state is obtained through measurement outcomes. A quantum measurement with \(M\) outcomes is described by a positive operator-valued measure (POVM)
\(
\mathcal{E}=\{\mathbf{E}_m\}_{m=1}^M
\)
on \(\mathcal{H}_q\), where each effect
\(
\mathbf{E}_m\in\mathbb{C}^{Q\times Q}
\)
is an operator in the Hilbrt space, which satisfies
\begin{equation}
\label{eq:povm_effects}
\mathbf{E}_m \succeq \mathbf{0}, \qquad \sum_{m=1}^M \mathbf{E}_m = \mathbf{I}.
\end{equation}
If the system is in state \(\bm{\rho}\), then the probability of observing outcome \(m \in \{ 1, \ldots, M\}\) is given by the Born rule,
\begin{equation}
\label{eq:povm_born_rule}
\pi_m(\bm{\rho})=\operatorname{Tr}(\mathbf{E}_m\bm{\rho}), \qquad m=1,\dots,M.
\end{equation}
Thus, the state must be inferred from measurement statistics. Thus, at a given measurement trial $t$, a categorical outcome $y_t \in \{1,\dots,M\}$ is observed with 
\[
y_t \sim \mathrm{Categorical}(\boldsymbol{\pi}(\bm{\rho}))
\]
where  
\[
\bm{\pi}(\bm{\rho}) := [\pi_1(\bm{\rho}), \ldots, \pi_M(\bm{\rho})]^\top.
\]
For notational brevity, consider the one-hot vector representation for measurements
$\mathbf{y}_t \in \{0,1\}^M$, defined as
\[
y_{t,m} = \mathbb{I}\{y_t = m\}, \quad \sum_{m=1}^M y_{t,m} = 1.
\]
Under this representation, for observation $\mathbf{y}_t$ it holds that  
\[\mathbb{E}[\mathbf{y}_t] = \boldsymbol{\pi}(\bm{\rho}), \qquad t = 1, \ldots T. \] 
Given \(T\) trials (or \emph{shots}) of stochastic measurements $\{\mathbf{y}_t\}_{t=1}^T$, the objective is to estimate the quantum the state $\bm{\rho}$. To this end, consider the empirical $m$-th outcome distribution
\begin{equation}
\label{eq:empirical_frequencies}
\hat{\pi}_m:=\frac{1}{T} \sum_{t=1}^T \mathbb{I}(y_{t,m} = 1), \qquad m=1,\dots,M.
\end{equation}
and let \(
\hat{\bm{\pi}} := [\hat{\pi}_1  \dots \hat{\pi}_M]^\top
\) be the empirical outcome distribution. A quantum state can then be estimated as 
\begin{equation}
\label{eq:state_estimation_problem}
\hat{\bm{\rho}}\in
\underset{\bm{\rho} \in \mathcal{D}(\mathcal{H}_q)}{\argmin}\;
D\bigl(\hat{\bm{\pi}},\bm{\pi}(\bm{\rho})\bigr).
\end{equation}
where $\mathcal{D}(\mathcal{H}_q)$ represents the feasible set for physical state space of \(q\)-qubit systems  
\begin{equation}
\label{eq:physical_state_space}
\mathcal{D}(\mathcal{H}_q):=
\left\{
\bm{\rho}\in\mathbb{C}^{Q\times Q}:
\bm{\rho}\succeq \mathbf{0},\ 
\operatorname{Tr}(\bm{\rho})=1
\right\},
\end{equation}
and \(D\) in \eqref{eq:state_estimation_problem}  is an appropriate discrepancy measure; e.g., least-squares error or KL-divergence, among others. As will be clarified in~Sec.~\ref{sec:multi_region}, our objective here is full-state estimation over the entire mixed-state model \(\mathcal{D}(\mathcal{H}_q)\), rather than estimation of a restricted class of observables or states. 

In general, the map
\(
\bm{\rho}\mapsto\bm{\pi}(\bm{\rho})
\)
need not be injective on \(\mathcal{D}(\mathcal{H}_q)\), so distinct states may produce the same measurement statistics, yielding an under-determined problem in \eqref{eq:state_estimation_problem}. To guarantee injectivity, the POVM \(\mathcal{E}\) must be informationally complete, meaning
\begin{equation}
\label{eq:informational_completeness_span}
\operatorname{span}(\mathcal{E})
=
\mathrm{Herm}(\mathcal{H}_q)
\end{equation}
where \(\mathrm{Herm}(\mathcal{H}_q)\) denotes the real vector space of Hermitian operators on \(\mathcal{H}_q\). Since
\(
\dim\!\bigl(\mathrm{Herm}(\mathcal{H}_q)\bigr)=Q^2,
\)
this requires
\(
M\ge Q^2.
\)
Under \eqref{eq:informational_completeness_span}, the map
\(
\bm{\rho}\mapsto\bm{\pi}(\bm{\rho})
\)
is injective on \(\mathcal{D}(\mathcal{H}_q)\), so the measurement statistics determine the underlying state uniquely in principle \cite{NielsenChuang,Watrous}.

The estimation problem becomes more delicate in the presence of readout errors. Let
\(
y^\star\in\{1,\dots,M\}
\)
and
\(
y\in\{1,\dots,M\}
\)
denote the ideal and recorded measurement outcomes, respectively. We model the readout process by a confusion matrix
\(
\mathbf{C}\in\mathbb{R}^{M\times M},
\)
defined by
\begin{equation}
\label{eq:confusion_matrix}
[\mathbf{C}]_{m,m'}:=\Pr(y=m\mid y^\star=m'),
\qquad
m,m'=1,\dots,M.
\end{equation}
Accordingly,
\begin{equation}
\label{eq:confusion_matrix_constraints}
0\le [\mathbf{C}]_{m,m'}\le 1,
\quad
\sum_{m=1}^{M} [\mathbf{C}]_{m,m'} =1,
\quad
m'=1,\dots,M,
\end{equation}
so each column of \(\mathbf{C}\) is a conditional distribution over recorded outcomes given the corresponding ideal outcome. If
\(
\bm{\pi}(\bm{\rho})
\)
denotes the ideal Born distribution associated with \(\bm{\rho}\), then the corresponding noisy outcome distribution is
\begin{equation}
\label{eq:noisy_distribution}
\bm{\pi}_{\mathbf{C}}(\bm{\rho})=\mathbf{C}\,\bm{\pi}(\bm{\rho}),
\end{equation}
that is,
\begin{equation}
\label{eq:noisy_distribution_components}
[\bm{\pi}_{\mathbf{C}}(\bm{\rho})]_m
=
\sum_{m'=1}^{M} [\mathbf{C}]_{m,m'} \,\pi_{m'}(\bm{\rho}),
\qquad
m=1,\dots,M.
\end{equation}
The confusion-matrix model captures a broad and practically relevant class of assignment errors, including both independent and correlated misclassification among measurement outcomes \cite{Maciejewski2020,Nation2021,Aasen2024}.

Incorporating measurement errors via $\mathbf{C}$, leads to the joint estimation problem over
\(
\mathcal{D}(\mathcal{H}_q)\times\mathcal{C}
\) 
\begin{equation}
\label{eq:noisy_state_estimation_problem}
(\hat{\bm{\rho}},\hat{\mathbf{C}})
\in
\underset{\bm{\rho},\,\mathbf{C}}{\argmin}\;
D \bigl(\hat{\bm{\pi}},\bm{\pi}_{\mathbf{C}}(\bm{\rho})\bigr),
\end{equation}
with the feasible set of confusion matrices defined as
\begin{equation}
\label{eq:admissible_confusion_matrices}
\mathcal{C}
:=
\left\{
\mathbf{C}\in\mathbb{R}^{M\times M}:
[\mathbf{C}]_{m m'}\ge 0,\ 
\mathbf{1}^\top \mathbf{C}=\mathbf{1}^\top
\right\}
\end{equation} 
possibly with extra structural or calibration constraints.

A more fundamental limitation of \eqref{eq:noisy_state_estimation_problem} is its lack of scalability as \(q\) increases. Under informational completeness, one has \(M\ge Q^2\), so the outcome space itself grows exponentially in the number of qubits. Consequently,
\(
\hat{\bm{\pi}}\in\mathbb{R}^M
\),
\(
\bm{\pi}_{\mathbf{C}}(\bm{\rho})\in\mathbb{R}^M
\),
and
\(
\mathbf{C}\in\mathbb{R}^{M\times M}
\)
are all exponentially large objects. In particular, under the column-stochasticity constraints, \(\mathbf{C}\) contains \(M(M-1)\) free parameters. The statistical burden is comparably severe. Since \(\hat{\bm{\pi}}\) is the empirical distribution of an \(M\)-outcome experiment based on \(T\) shots, standard concentration bounds imply that maintaining a fixed estimation accuracy as \(M\) grows requires \(T\) to increase at least linearly with \(M\), and typically on the order of \(M/\varepsilon^2\) to achieve accuracy \(\varepsilon\) in \(\ell_1\)-type metrics \cite{Weissman2003}. Therefore, the global formulation \eqref{eq:noisy_state_estimation_problem} quickly becomes impractical at large \(q\), which motivates the development of formulations that exploit additional structure.

\section{Multi-Region Formulation}
\label{sec:multi_region}

The global formulation \eqref{eq:noisy_state_estimation_problem} becomes computationally and statistically prohibitive as the number of qubits $q$ increases. This motivates  decomposing the full system into smaller, possibly overlapping regions, and working with regional objects rather than a single global pair \((\bm{\rho},\mathbf{C})\); as illustrated on top of Fig.~\ref{fig:finalq}. Related ideas appear in quantum overlapping tomography and its recent extensions, as well as in compatibility-constrained reconstruction from overlapping reduced states \cite{Cotler2020,Yang2023,Hu2024,Wei2025,Wang2025}. This perspective is adapted here to joint state estimation with readout errors, where regional density operators and regional confusion matrices are coupled through consistency constraints on shared subsystems. 

\begin{figure}[t]
    \centering
    \includegraphics[width=\columnwidth]{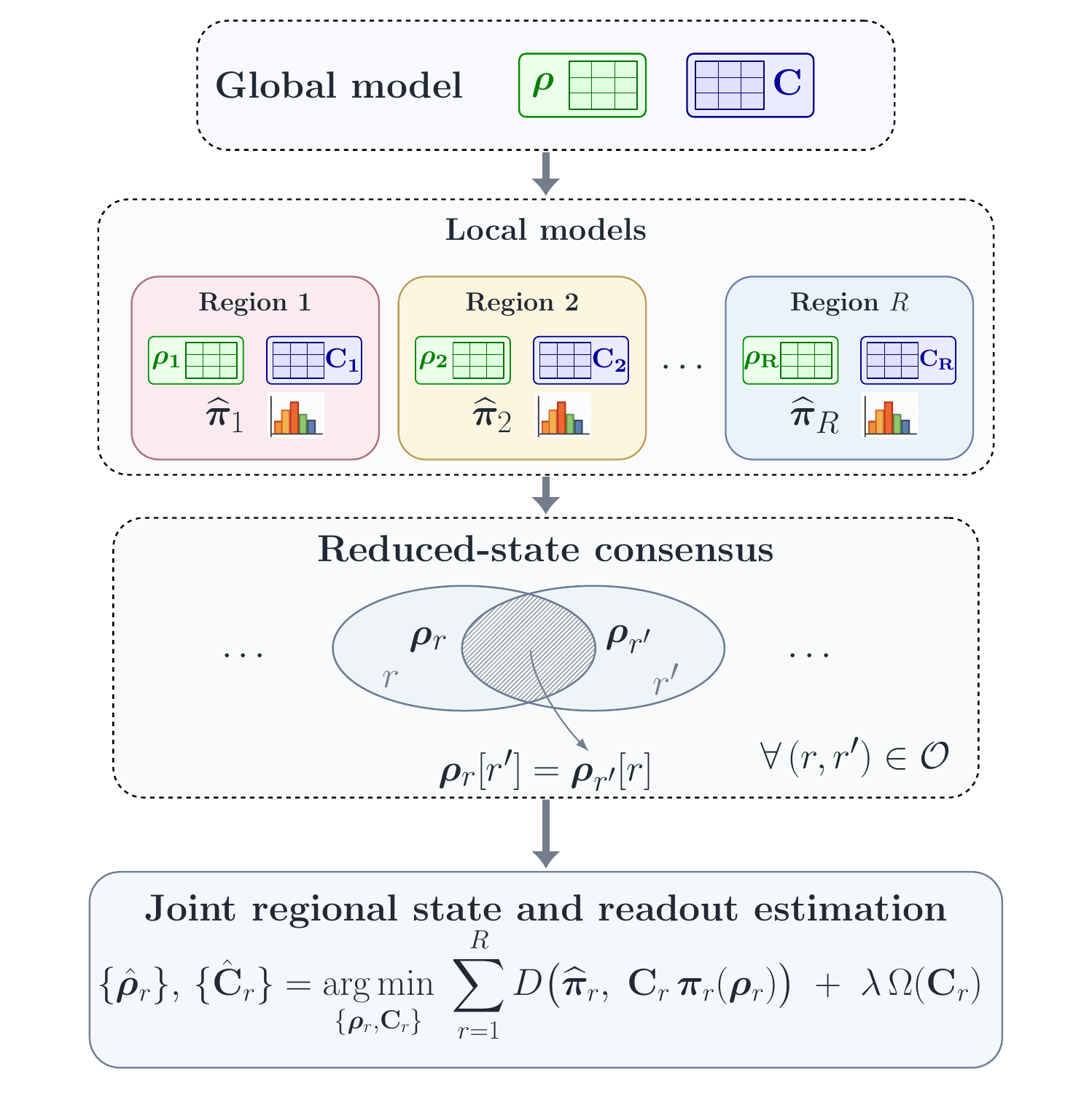}
    \caption{A global quantum system partitioned into $R$ local regions. Reduced-state consensus among overlapping regions is reached by partial trace over regional states, offering distributed joint regional state and readout estimation.}
    \label{fig:finalq}
\end{figure}

\subsection{Global-to-Local Decomposition with Overlap Consensus}

Let \(\bm{\rho}\in\mathcal{D}(\mathcal{H}_q)\) denote the global \(q\)-qubit state introduced in Sec.~\ref{sec:basics_background}. To obtain a structured alternative to the global problem \eqref{eq:noisy_state_estimation_problem}, suppose that the full system is distributed over \(s=1,\dots,S\) sites, where site \(s\) contains \(q_s\) qubits, so that
\(
q=\sum_{s=1}^S q_s
\).
Next, introduce \(R\) possibly overlapping regions, indexed by \(r=1,\dots,R\), and let
\(
\mathcal{R}_r\subseteq\{1,\dots,S\}
\)
denote the set of sites contained in region \(r\). The total number of qubits in region \(r\) is then
\(
q_r:=\sum_{s\in\mathcal{R}_r} q_s
\),
and the corresponding regional state is represented by a density operator
\(
\bm{\rho}_r\in\mathcal{D}(\mathcal{H}_{q_r})
\). For brevity, let
\(
\bm{\rho}_r\in\mathcal{D}_r
\)
denote the regional density operator, where \(\mathcal{D}_r\) is the regional counterpart of \eqref{eq:physical_state_space}. Let
\(
\mathcal{E}_r := \{\mathbf{E}_{r,m}\}_{m=1}^{M_r}
\)
be a POVM on \(\mathcal{H}_{q_r}\), and let
\(
\bm{\pi}_r(\bm{\rho}_r)
:=
[\pi_{r,1}(\bm{\rho}_r),\dots,\pi_{r,M_r}(\bm{\rho}_r)]^\top
\)
denote the associated ideal outcome distribution, where
\(
\pi_{r,m}(\bm{\rho}_r) := \operatorname{Tr}(\mathbf{E}_{r,m}\bm{\rho}_r)
\).
Readout errors are modeled by a regional confusion matrix
\(
\mathbf{C}_r\in\mathcal{C}_r
\subseteq
\mathbb{R}^{M_r\times M_r},
\)
where \(\mathcal{C}_r\) is the regional counterpart of \eqref{eq:admissible_confusion_matrices}. The corresponding noisy regional outcome distribution is
\(
\bm{\pi}_{\mathbf{C}_r}(\bm{\rho}_r)=\mathbf{C}_r\bm{\pi}_r(\bm{\rho}_r)
\).

With the POVM \(\mathcal{E}_r\) repeated \(T_r\) times, let the empirical regional distribution as
\(
\hat{\bm{\pi}}_r=[\hat{\pi}_{r,1},\dots,\hat{\pi}_{r,M_r}]^\top
\),
where the empirical probability of 
\(m\)-th recorded outcome is 
\begin{align}
    \nonumber \hat{\pi}_{r,m}:=\frac{1}{T_r} \sum_{t=1}^{T_r} \mathbb{I}([\mathbf{y}_{rt}]_m = 1), \qquad m=1,\dots,M.
\end{align}
Thus, for each region \(r\), the observed distribution \(\hat{\bm{\pi}}_r\) is modeled by the noisy prediction \(\bm{\pi}_{\mathbf{C}_r}(\bm{\rho}_r)\). For overlapping regions, the relevant compatibility condition is that they determine the same reduced state on their intersection. Let
\begin{equation}
\label{eq:overlap_pairs}
\mathcal{O}:=\{(r,r'):\mathcal{R}_r\cap\mathcal{R}_{r'}\neq\varnothing\}.
\end{equation}

\begin{definition}[Reduced state on an overlap {\rm\cite[p.~80]{ref:partial_trace}}]
\label{def:reduced_overlap_state}
For \( \forall (r,r')\in\mathcal{O}\), write
\[
\mathcal{H}_{q_r}\cong \mathcal{H}_{r\cap r'}\otimes \mathcal{H}_{r\setminus r'}
\]
where \(\mathcal{H}_{r\cap r'}\) and \(\mathcal{H}_{r\setminus r'}\) are the Hilbert spaces associated with \(\mathcal{R}_r\cap\mathcal{R}_{r'}\) and \(\mathcal{R}_r\setminus\mathcal{R}_{r'}\), respectively. For any \(\bm{\rho}_r\in\mathcal{D}(\mathcal{H}_{q_r})\), define its reduced state on the shared subsystem by the basis-independent partial trace
\begin{equation}
\label{eq:overlap_reduced_state}
\bm{\rho}_r[r'] 
:=\operatorname{Tr}_{\mathcal{R}_r\setminus\mathcal{R}_{r'}}(\bm{\rho}_r)
=\!\! \sum_{j=1}^{N_{r\setminus r'}} \!
\bigl(\mathbf{I}_{r\cap r'}\otimes \langle \bm{\phi}_j|\bigr)\bm{\rho}_r
\bigl(\mathbf{I}_{r\cap r'}\otimes |\bm{\phi}_j\rangle\bigr)
\end{equation}
where \(\{|\bm{\phi}_j\rangle\}_{j=1}^{N_{r\setminus r'}}\) is any orthonormal basis of \(\mathcal{H}_{r\setminus r'}\); \(I_{r\cap r'}\) is the identity on \(\mathcal{H}_{r\cap r'}\); and \(N_{r\setminus r'}:=\dim(\mathcal{H}_{r\setminus r'})\).
\end{definition}
Since \(\mathcal{R}_r\setminus\mathcal{R}_{r'}=\mathcal{R}_r\setminus(\mathcal{R}_r\cap\mathcal{R}_{r'})\), the partial trace removes exactly the sites of \(r\) outside the overlap. Thus, overlapping regions \(r\) and \(r'\) satisfy a consensus constraint; that is,  
\begin{equation}
\label{eq:overlap_consistency}
\bm{\rho}_r [r'] = \bm{\rho}_{r'} [r], \qquad (r,r')\in\mathcal{O}
\end{equation}
with both sides density operators on the subsystem associated with \(\mathcal{R}_r\cap\mathcal{R}_{r'}\); see also  Fig.~\ref{fig:finalq} for the system model.

\subsection{Joint State and Confusion Matrix Estimation}
Throughout this subsection, whenever two region indices \(r\) and \(r'\) appear together, they are understood to refer to an overlapping pair in \(\mathcal{O}\) in \eqref{eq:overlap_pairs}.

The corresponding multi-region counterpart of \eqref{eq:noisy_state_estimation_problem} jointly estimates the regional states \(\{\bm{\rho}_r\}\) and regional confusion matrices \(\{\mathbf{C}_r\}\) by fitting the predicted noisy regional distributions to the empirical distributions, while enforcing physicality, column-stochasticity, and overlap consensus. This yields
\begin{subequations}
\label{eq:multi_region_estimation_problem}
\begin{align}
(\{\hat{\bm{\rho}}_r\},\{\hat{\mathbf{C}}_r\})
&\in
\notag\\
\MoveEqLeft[5.5]
\argmin_{\mathclap{\{\bm{\rho}_r\},\,\{\mathbf{C}_r\}}}\;
\sum_{r=1}^R
\Bigl[
D\bigl(\hat{\bm{\pi}}_r,\mathbf{C}_r\bm{\pi}_r(\bm{\rho}_r)\bigr)
+\lambda\,\Omega_r(\mathbf{C}_r)
\Bigr]
\label{eq:multi_region_estimation_problem_obj}
\\[0.4ex]
\mathrlap{\hspace{-0.9em}\text{s.t.}\quad}\phantom{\text{s.t.}\quad}
\bm{\rho}_r
&\in
\mathcal{D}_r,
\qquad  r = 1, \ldots, R,
\label{eq:multi_region_estimation_problem_state}
\\
\phantom{\text{s.t.}\quad}
\mathbf{C}_r
&\in
\mathcal{C}_r,
\qquad \; r= 1, \ldots, R,
\label{eq:multi_region_estimation_problem_confusion}
\\
\phantom{\text{s.t.}\quad}
\bm{\rho}_{r'}[r] &= \bm{\rho}_{r}[r']
\quad \;\; \forall \, (r, r') \in \mathcal{O}.
\label{eq:multi_region_estimation_problem_overlap}
\end{align}
\end{subequations}
Here \(\lambda\ge 0\) is a regularization parameter, and \(\Omega_r(\mathbf{C}_r)\) is a penalty on the regional confusion matrix. The objective in \eqref{eq:multi_region_estimation_problem_obj} fits the empirical regional distributions \(\hat{\bm{\pi}}_r\) to the predicted noisy distributions \(\mathbf{C}_r\bm{\pi}_r(\bm{\rho}_r)\), while \eqref{eq:multi_region_estimation_problem_state}, \eqref{eq:multi_region_estimation_problem_confusion}, and \eqref{eq:multi_region_estimation_problem_overlap} enforce the regional counterparts of the state constraint \eqref{eq:physical_state_space}, the confusion-matrix constraint \eqref{eq:admissible_confusion_matrices}, and the overlap-consensus constraint \eqref{eq:overlap_consistency}. 
The cost in \eqref{eq:multi_region_estimation_problem_obj}
includes the regularizer
\begin{equation}
\label{eq:regional_confusion_regularizer_general}
\Omega_r(\mathbf{C}_r)=\|\mathbf{C}_r-\bar{\mathbf{C}}_r\|_F^2
\end{equation}
where \(\bar{\mathbf{C}}_r\) is a prescribed reference confusion matrix; when no additional calibration information is available, a natural choice is \(\bar{\mathbf{C}}_r=\mathbf{I}_r\), corresponding to ideal readout.

Problem~\eqref{eq:multi_region_estimation_problem} is not jointly convex because \(\mathbf{C}_r\bm{\pi}_r(\bm{\rho}_r)\) is bilinear in \((\bm{\rho}_r,\mathbf{C}_r)\). For fixed \(\{\mathbf{C}_r\}\) however, it is convex in \(\{\bm{\rho}_r\}\) because \(\bm{\pi}_r(\bm{\rho}_r)\) and the partial-trace maps underlying \eqref{eq:overlap_consistency} are linear. Likewise, for fixed \(\{\bm{\rho}_r\}\), it is convex in \(\{\mathbf{C}_r\}\) provided \(D(\hat{\bm{\pi}}_r,\cdot)\) and \(\Omega_r\) are convex. Thus, \eqref{eq:multi_region_estimation_problem} has a blockwise convex structure, where the state block is coupled through the overlap constraints, whereas the confusion-matrix block is separable across regions.


\subsection{Distributed Solution via ADMM}

Problem~\eqref{eq:multi_region_estimation_problem} is separable across regions except for the overlap constraints \eqref{eq:multi_region_estimation_problem_overlap}. We therefore use a proximal alternating iterative scheme in which, over outer iterations $k = 1, \ldots K$ with a  fixed \(\{\mathbf{C}_r^k\}\), the state block is solved by ADMM across overlapping regions, and for fixed \(\{\bm{\rho}_r^{k+1}\}\), the confusion matrices are updated locally in parallel.

For each overlapping pair, introduce an auxiliary variable \(\bm{\rho}_{rr'}\) on the shared subsystem, with the convention
$
\bm{\rho}_{rr'}=\bm{\rho}_{r'r}.
$
Then, the overlap-consistency constraints can be written as
\begin{equation}
\label{eq:overlap_consensus_constraints}
\bm{\rho}_r[r'] = \bm{\rho}_{rr'}
\qquad
\bm{\rho}_{r'}[r] = \bm{\rho}_{r'r}.
\end{equation}
At outer iteration \(k+1\), with \(\{\mathbf{C}_r^k\}\) fixed and \(\bm{\rho}_r\in\mathcal{D}_r\), the proximal state subproblem for \(r = 1, \ldots R\) is
\begin{equation}
\label{eq:proximal_state_subproblem_overlap}
\begin{aligned}
(\{\bm{\rho}_r^{k+1}\},\{\bm{\rho}_{rr'}^{k+1}\})
&\in
\\
\MoveEqLeft[7.5]
\argmin_{\mathclap{\{\bm{\rho}_r\},\{\bm{\rho}_{rr'}\}}}\;
\sum_{r=1}^R
\left[
D\left(\hat{\bm{\pi}}_r,\mathbf{C}_r^k\bm{\pi}_r(\bm{\rho}_r)\right)
+
\frac{\gamma_\rho}{2}\|\bm{\rho}_r-\bm{\rho}_r^k\|_F^2
\right]
\\
\mathrlap{\hspace{-4.3em}\text{s.t.}\quad}\phantom{\text{s.t.}\quad}
&\bm{\rho}_r[r'] = \bm{\rho}_{rr'},
\quad
\bm{\rho}_{r'}[r] = \bm{\rho}_{r'r}.
\end{aligned}
\end{equation}
For fixed \(\{\mathbf{C}_r^k\}\), problem~\eqref{eq:proximal_state_subproblem_overlap} is convex whenever \(D(\hat{\bm{\pi}}_r,\cdot)\) is convex, since \(\bm{\pi}_r(\cdot)\) and the partial-trace maps $\bm{\rho}_r[r']$ and $\bm{\rho}_{r'}[r]$ are linear.

To simplify the notation in the ADMM step, define
\begin{equation}
\label{eq:state_loss_definition}
f_r^k(\bm{\rho}_r)
:=
D\left(\hat{\bm{\pi}}_r,\mathbf{C}_r^k\bm{\pi}_r(\bm{\rho}_r)\right)
+
\frac{\gamma_\rho}{2}\|\bm{\rho}_r-\bm{\rho}_r^k\|_F^2.
\end{equation}
Then, for fixed outer iterate \(k\), the augmented Lagrangian is
\begin{align}
&\mathcal{L}_\beta^k\!\left(
\{\bm{\rho}_r\},\{\bm{\rho}_{rr'}\};\{\bm{\Lambda}_{rr'}\}
\right) := \sum_{r=1}^R f_r^k(\bm{\rho}_r)
\label{eq:augmented_lagrangian_state_subproblem}
\\
\nonumber
& \quad 
+ \sum_{(r,r')\in\mathcal{O}}
\Big\{
\langle \bm{\Lambda}_{rr'},\bm{\rho}_r[r'] - \bm{\rho}_{rr'} \rangle
+\frac{\beta}{2}\|\bm{\rho}_{r}[r'] - \bm{\rho}_{rr'}\|_F^2
\\
\nonumber
& \qquad + \langle \bm{\Lambda}_{r'r},\bm{\rho}_{r'}[r] - \bm{\rho}_{r'r}\rangle
+\frac{\beta}{2}\|\bm{\rho}_{r'}[r] - \bm{\rho}_{r'r}\|_F^2
\Big\}.
\end{align}
For a fixed outer iteration \(k\), we suppress the superscript \(k\) throughout the inner ADMM iteration indexed by $l = 1, \ldots $. 
The corresponding inner ADMM iterations are
\begin{subequations}
\label{eq:admm_generic_updates}
\begin{align}
\{\bm{\rho}_r^{\ell+1}\}
:=& \,
\underset{\{\bm{\rho}_r\}}{\arg\min}\;
\mathcal{L}_\beta^k\left(
\{\bm{\rho}_r\},
\{\bm{\rho}_{rr'}^\ell\};
\{\bm{\Lambda}_{rr'}^\ell\}
\right),
\label{eq:admm_generic_state_update}
\\
\{\bm{\rho}_{rr'}^{\ell+1}\}
:=& \,
\underset{\{\bm{\rho}_{rr'}\}}{\arg\min}\;
\mathcal{L}_\beta^k\left(
\{\bm{\rho}_r^{\ell+1}\},
\{\bm{\rho}_{rr'}\};
\{\bm{\Lambda}_{rr'}^\ell\}
\right),
\label{eq:admm_generic_consensus_update}
\\
\bm{\Lambda}_{rr'}^{\ell+1} = & 
\mathrlap{\hspace{0.15em}{} {}\bm{\Lambda}_{rr'}^\ell+\beta\,(\bm{\rho}_{r}^{\ell+1}[r'] - \bm{\rho}_{rr'}^{\ell+1}),}
\phantom{{} = \,{}\bm{\Lambda}_{rr'}^\ell+\beta \,(\bm{\rho}_{r}^{\ell+1}[r'] - \bm{\rho}_{rr'}^{\ell+1}),}
\\
\bm{\Lambda}_{r'r}^{\ell+1}
=& \,
\bm{\Lambda}_{r'r}^\ell+\beta\,(\bm{\rho}_{r'}^{\ell+1}[r] - \bm{\rho}_{r'r}^{\ell+1}).
\label{eq:admm_generic_dual_update}
\end{align}
\end{subequations}
The next proposition gives the equivalent explicit regional updates used in the algorithm; see the Appendix for the proof. 

\begin{proposition}
\label{prop:explicit_admm_updates}
For a fixed outer iterate \(k\), the ADMM steps in \eqref{eq:admm_generic_updates} are equivalent to the following explicit updates. For each fixed region \(r\), the minimization is over \(\mathcal{D}_r\), and the sum over \(r'\) runs over all regions overlapping with \(r\) 
\begin{subequations}
\label{eq:admm_explicit_updates}
\begin{align}
&\bm{\rho}_r^{\ell+1}\in \argmin_{\bm{\rho}_r}\;
\bigg\{
f_r^k(\bm{\rho}_r) \, + \sum_{r'}
\Big[
\langle \bm{\Lambda}_{rr'}^\ell,\bm{\rho}_{r}[r'] -\bm{\rho}_{rr'}\rangle
\nonumber \\
& \hspace{+3.5 cm}
+\frac{\beta}{2}\|\bm{\rho}_{r}[r'] -\bm{\rho}_{rr'}\|_F^2
\big]
\bigg\},
\label{eq:admm_state_update}
\\
& \bm{\rho}_{rr'}^{\ell+1}
=
\frac{1}{2}
\left(
\bm{\rho}_{r}^{\ell+1}[r']
+
\bm{\rho}_{r'}^{\ell+1}[r]
+
\frac{\bm{\Lambda}_{rr'}^\ell+\bm{\Lambda}_{r'r}^\ell}{\beta}\right),
\label{eq:admm_consensus_update}
\\
& \bm{\Lambda}_{rr'}^{\ell+1}
=
\bm{\Lambda}_{rr'}^\ell
+
\beta
\left(
\bm{\rho}_{r}^{\ell+1}[r']
-
\bm{\rho}_{rr'}^{\ell+1}
\right),
\label{eq:admm_dual_update_r}
\\
& \bm{\Lambda}_{r'r}^{\ell+1}
=
\bm{\Lambda}_{r'r}^\ell
+
\beta
\left(
\bm{\rho}^{\ell+1}_{r'}[r]
-
\bm{\rho}_{r'r}^{\ell+1}
\right).
\label{eq:admm_dual_update_rprime}
\end{align}
\end{subequations}
\end{proposition}

The inner ADMM loop is run until convergence. Let \(\ell_\star\) denote the terminal inner iterate. We then set
\[
\bm{\rho}_r^{k+1}:=\bm{\rho}_r^{\ell_\star},
\qquad
\bm{\rho}_{rr'}^{k+1}:=\bm{\rho}_{rr'}^{\ell_\star},
\qquad
\bm{\Lambda}_{rr'}^{k+1}:=\bm{\Lambda}_{rr'}^{\ell_\star}.
\]
The only exchanged quantities are the reduced states on shared subsystems; no global density operator is formed.

With \(\{\bm{\rho}_r^{k+1}\}\) fixed, each region $r$ updates its confusion matrix over \(\mathcal{C}_r\) by solving
\begin{equation}
\label{eq:admm_confusion_update}
\begin{aligned}
\mathrlap{\hspace{-2.3em}\mathbf{C}_r^{k+1}\in}\phantom{\mathbf{C}_r^{k+1}\in}
&
\\
\MoveEqLeft[5.5]
\argmin_{\mathbf{C}_r}\;
\begin{alignedat}[t]{1}
[\,D  \left(\hat{\bm{\pi}}_r,\mathbf{C}_r\bm{\pi}_r(\bm{\rho}_r^{k+1})\right)
&{}+ \lambda\,\Omega_r(\mathbf{C}_r)
\\
&{}+ \frac{\gamma_C}{2}\|\mathbf{C}_r-\mathbf{C}_r^k\|_F^2\,].
\end{alignedat}
\end{aligned}
\end{equation}
These updates are separable across regions and require no inter-region communication. No global confusion matrix is formed.

\begin{algorithm}[t]
\caption{D-ADMM for Robust Quantum State Estimation}
\label{alg:distributed_paadmm}
\begin{algorithmic}[1]
\REQUIRE \(\{\mathbf{y}_{t,r}\}, \{\mathcal{E}_r\}\), and positive 
\(\beta\), \(\gamma_\rho\), \(\gamma_C\)
\RETURN{\(\{\bm{\rho}_r, \mathbf{C}_r\}_{r=1}^R\)}
\STATE{{Initialize:~}\(\{\bm{\rho}_r^0,\mathbf{C}_r^0\}\),
\(\{\bm{\rho}_{rr'}^0\}\),
\(\{\bm{\Lambda}_{rr'}^0\}\),
}
\FOR{$k=0,1,2,\dots$} 
    \STATE Initialize inner loop: \(\ell= 0, \bm{\rho}_{rr'}^{0}=\bm{\rho}_{rr'}^{k}, \bm{\Lambda}_{rr'}^{0}=\bm{\Lambda}_{rr'}^{k}\)
    \REPEAT
        \FOR{each region $r=1,\dots,R$ \textbf{in parallel}}
            \STATE Receive \(\{\bm{\rho}_{rr'}^\ell,\bm{\Lambda}_{rr'}^\ell\}_{r'}\), \( \forall r'~\textrm{~s.t.~} (r,r') \in \mathcal{O}\)
            \STATE Compute \(\bm{\rho}_r^{\ell+1}\) via \eqref{eq:admm_state_update}
            \STATE Send \(\{\bm{\rho}_{r}[r']^{\ell+1}\}_{r'}\) to \(\forall r', \textrm{\,~~s.t.~}(r,r') \in \mathcal{O}\).
        \ENDFOR
        \FORALL{overlapping pairs \((r,r')\) \textbf{in parallel}}
            \STATE Compute \(\bm{\rho}_{rr'}^{\ell+1}\) via \eqref{eq:admm_consensus_update}
            \STATE Update \; \(\bm{\Lambda}_{rr'}^{\ell+1}\) via \eqref{eq:admm_dual_update_r}
            \STATE Update \; \(\bm{\Lambda}_{r'r}^{\ell+1}\) via \eqref{eq:admm_dual_update_rprime}
        \ENDFOR
        \STATE \(\ell = \ell+1\)
    \UNTIL{The inner ADMM stopping criterion is satisfied}
    \STATE{Initialize outer loop: \newline \(\ell_\star = \ell, \bm{\rho}_r^{k+1} = \bm{\rho}_r^{\ell_\star}, \bm{\rho}_{rr'}^{k+1} = \bm{\rho}_{rr'}^{\ell_\star}, \textrm{~and~} \bm{\Lambda}_{rr'}^{k+1} =  \bm{\Lambda}_{rr'}^{\ell_\star}\)}
    \FOR{each region $r=1,\dots,R$ \textbf{in parallel}}
        \STATE Compute~~~ \(\mathbf{C}_r^{k+1}\) via \eqref{eq:admm_confusion_update}
    \ENDFOR
    \IF{The outer iteration stopping criterion is satisfied}
        \STATE \textbf{Return~}\(\{\bm{\rho}_r, \mathbf{C}_r\}_{r=1}^R\)
    \ENDIF
\ENDFOR
\end{algorithmic}
\end{algorithm}
\begin{remark}[KL-based discrepancy and maximum likelihood estimation]
\label{rem:kl_ml_map}
Suppose that in \eqref{eq:multi_region_estimation_problem_obj} the discrepancy measure is the Kullback-Leibler (KL) divergence; that is, \(D\bigl(\hat{\bm{\pi}}_r,\mathbf C_r\bm{\pi}_r(\bm{\rho}_r)\bigr)=\mathrm{KL} \bigl(\hat{\bm{\pi}}_r \,\|\, \mathbf C_r\bm{\pi}_r(\bm{\rho}_r)\bigr), \forall r\). Then \eqref{eq:multi_region_estimation_problem} (without the regularizer) is equivalent to a constrained maximum-likelihood estimator (MLE). With a properly defined regularizer as in \eqref{eq:multi_region_estimation_problem_obj}, the problem becomes a constrained penalized MLE; equivalently, it is a constrained maximum a posteriori (MAP) estimator under a regularizer satisfying log prior \(-\log p\bigl(\{\bm{\rho}_r\},\{\mathbf C_r\}\bigr)\); see also  Appendix~\ref{app:proof_kl_ml_map} for details.
\end{remark}


\section{Analytical Performance}
\label{sec:theory}
\subsection{Local identifiability and quadratic growth}

The local behavior of \eqref{eq:multi_region_estimation_problem} will be investigated here around a feasible interior reference pair
\(
(\{\bm{\rho}_r^\ast\},\{\mathbf{C}_r^\ast\})
\),
where
\(
\bm{\rho}_r^\ast\in\mathcal{D}_r
\)
and
\(
\mathbf{C}_r^\ast\in\mathcal{C}_r~\forall r\), with the interior conditions
\[
\bm{\rho}_r^\ast\succ0,
\qquad
[\mathbf{C}_r^\ast]_{mm'}>0,
\qquad
\forall r,\ \forall m,m'.
\]
These are the strict counterparts of the regional feasibility constraints in
\eqref{eq:multi_region_estimation_problem_state}--\eqref{eq:multi_region_estimation_problem_confusion}. For any feasible pair
\(
(\{\bm{\rho}_r\},\{\mathbf{C}_r\})
\),
define the distance of pair \((\{\bm{\rho}_r\},\{\mathbf{C}_r\})\) to \((\{\bm{\rho}_r^\ast\},\{\mathbf{C}_r^\ast\})\) as 
\begin{equation}
\label{eq:local_distance_ast}
d_\ast^2\bigl(\{\bm{\rho}_r\},\{\mathbf{C}_r\}\bigr)
:=
\sum_{r=1}^R
\Bigl(
\|\bm{\rho}_r-\bm{\rho}_r^\ast\|_F^2
+
\|\mathbf{C}_r-\mathbf{C}_r^\ast\|_F^2
\Bigr)
\end{equation}
with the quadratic growth penalty given by  
\begin{equation}
\label{eq:local_population_misfit_ast}
\mathcal{Q}_\ast\bigl(\{\bm{\rho}_r\},\{\mathbf{C}_r\}\bigr)
:=
\frac12
\sum_{r=1}^R
\left\|
\mathbf{C}_r\bm{\pi}_r(\bm{\rho}_r)
-
\mathbf{C}_r^\ast\bm{\pi}_r(\bm{\rho}_r^\ast)
\right\|_2^2 .
\end{equation}
which penalizes the observation model discrepancies under 
\((\{\bm{\rho}_r\},\{\mathbf{C}_r\})\) compared to those under 
\((\{\bm{\rho}_r^\ast\},\{\mathbf{C}_r^\ast\})\).

The local identifiability at a stationary point 
\(
(\{\bm{\rho}_r^\ast\},\{\mathbf{C}_r^\ast\})
\)
means that \( \exists \varepsilon>0\) such that \(\forall (\{\bm{\rho}_r\},\{\mathbf{C}_r\})\) with \(d_\ast^2(\{\bm{\rho}_r\},\{\mathbf{C}_r\})<\varepsilon\), if 
\(
\mathbf{C}_r\bm{\pi}_r(\bm{\rho}_r)=\mathbf{C}_r^\ast\bm{\pi}_r(\bm{\rho}_r^\ast) \, \forall r
\), then it holds that
\(
\bm{\rho}_r=\bm{\rho}_r^\ast
\)
and
\(
\mathbf{C}_r=\mathbf{C}_r^\ast \; \forall r\). This means a local stationary point is identifiable. Similarly, local quadratic growth guarantees that
\(
\mathcal{Q}_\ast\ge c\,d_\ast^2
\)
for some \(c>0\) and all feasible pairs with \(d_\ast\) sufficiently small. This implies that with a sufficiently small perturbation of a stationary point, the resulting solution incurs a quadratically larger penalty.

To state a sufficient condition for both properties, let
\(
\Delta\bm{\rho}_r:=\bm{\rho}_r-\bm{\rho}_r^\ast
\)
and
\(
\Delta\mathbf{C}_r:=\mathbf{C}_r-\mathbf{C}_r^\ast
\).
The linearized feasible perturbation set induced by
\eqref{eq:multi_region_estimation_problem_state}--\eqref{eq:multi_region_estimation_problem_overlap} is
\begin{align}
\mathcal{T}_\ast
&:= \left\{(\{\Delta\bm{\rho}_r\},\{\Delta\mathbf{C}_r\})\right\} \notag\\
&\text{s.t.}\quad
\Delta\bm{\rho}_r = \Delta\bm{\rho}_r^\dagger,\quad
\operatorname{Tr}(\Delta\bm{\rho}_r)=0,\quad
\mathbf{1}^\top\Delta\mathbf{C}_r=\mathbf{0}^\top, \notag\\
&\hphantom{\text{s.t.}\quad}
\Delta\bm{\rho}_r[r'] = \Delta\bm{\rho}_{r'}[r]
\quad \forall (r,r')\in\mathcal{O}.
\label{eq:linearized_feasible_set_ast}
\end{align}
where the term
\(
\bm{\pi}_r(\Delta\bm{\rho}_r)
\)
is understood from the regional counterpart of \eqref{eq:povm_born_rule}. Define also mapping
\begin{equation}
\label{eq:linearized_prediction_map_ast}
\mathcal{A}_\ast(\{\Delta\bm{\rho}_r\},\{\Delta\mathbf{C}_r\})
:=
\left\{
\Delta\mathbf{C}_r\,\bm{\pi}_r(\bm{\rho}_r^\ast)
+
\mathbf{C}_r^\ast\,\bm{\pi}_r(\Delta\bm{\rho}_r)
\right\}_{r=1}^R .
\end{equation}
The ensuing theorem shows that injectivity of \(\mathcal{A}_\ast\) on \(\mathcal{T}_\ast\) is sufficient for both local identifiability and local quadratic growth.

\begin{theorem}
\label{thm:local_identifiability}
Let
\(
(\{\bm{\rho}_r^\ast\},\{\mathbf{C}_r^\ast\})
\)
be a feasible interior reference pair for \eqref{eq:multi_region_estimation_problem}. Suppose that, for every
\(
(\{\Delta\bm{\rho}_r\},\{\Delta\mathbf{C}_r\})\in\mathcal{T}_\ast
\),
the implication
\(
\mathcal{A}_\ast(\{\Delta\bm{\rho}_r\},\{\Delta\mathbf{C}_r\})=0
\Rightarrow
\Delta\bm{\rho}_r=0,\ \Delta\mathbf{C}_r=0,\ \forall r
\)
holds. Then there exist small but nonzero constants \(\varepsilon>0\) and \(c>0\) so that every feasible pair
\(
(\{\bm{\rho}_r\},\{\mathbf{C}_r\})
\)
with \(d_\ast<\varepsilon\) satisfies
\(
\mathcal{Q}_\ast \ge c\,d_\ast^2
\);
in particular, if
\(
\mathbf{C}_r\bm{\pi}_r(\bm{\rho}_r)=\mathbf{C}_r^\ast\bm{\pi}_r(\bm{\rho}_r^\ast) \; \forall r\), then
\(
\bm{\rho}_r=\bm{\rho}_r^\ast
\)
and
\(
\mathbf{C}_r=\mathbf{C}_r^\ast \; \forall r
\).
\end{theorem}
\begin{remark}
\label{rem:injectivity_interpretation}
The assumption in Theorem~\ref{thm:local_identifiability} is precisely 
$
\ker(\mathcal{A}_\ast)\cap \mathcal{T}_\ast=\{(\mathbf{0},\mathbf{0})\}.
$ Equivalently, for every nonzero
\(
(\{\Delta\bm{\rho}_r\},\{\Delta\mathbf{C}_r\})\in\mathcal{T}_\ast
\),
there exists at least one region \(r\) such that
\(
\Delta\mathbf{C}_r\,\bm{\pi}_r(\bm{\rho}_r^\ast)
+
\mathbf{C}_r^\ast\,\bm{\pi}_r(\Delta\bm{\rho}_r)\neq 0
\).
That is, the linearization of
\(
(\{\bm{\rho}_r\},\{\mathbf{C}_r\}) \mapsto \{\mathbf{C}_r\bm{\pi}_r(\bm{\rho}_r)\}_{r=1}^R
\)
at
\(
(\{\bm{\rho}_r^\ast\},\{\mathbf{C}_r^\ast\})
\)
is injective on \(\mathcal{T}_\ast\).
\end{remark}

\subsection{Inner-ADMM convergence}

The inner ADMM loop is studied here to solve the state-update step in the distributed alternating scheme. Since the full problem \eqref{eq:multi_region_estimation_problem} is jointly nonconvex in
\(
(\{\bm{\rho}_r\},\{\mathbf{C}_r\})
\),
the result is stated for a fixed outer iterate \(k\), with the confusion matrices
\(
\{\mathbf{C}_r^k\}_{r=1}^R
\)
held fixed. The state-update step reduces then to the convex proximal state subproblem in \eqref{eq:proximal_state_subproblem_overlap}. Its objective is
\begin{align}
\label{eq:state_subproblem_objective}
& \mathcal J_{\rho}^{k}\!\bigl(\{\bm{\rho}_r\};\{\mathbf C_r\}\bigr) 
:= 
\nonumber 
\\
& \qquad \; \sum_{r=1}^R
\left[
D\bigl(\hat{\bm{\pi}}_r,\mathbf{C}_r^k\bm{\pi}_r(\bm{\rho}_r)\bigr)
+
\frac{\gamma_\rho}{2}\|\bm{\rho}_r-\bm{\rho}_r^k\|_F^2
\right].
\end{align}
For the inner ADMM iterates, define the overlap residuals
\begin{equation}
\label{eq:overlap_residual_definition_theorem}
\bm{\delta}_{rr'}^\ell
:=
\bm{\rho}^\ell_{r}[r']-\bm{\rho}_{rr'}^\ell,
\qquad
(r,r')\in\mathcal{O}.
\end{equation}
Primal feasibility of the inner loop requires that
\(
\bm{\delta}_{rr'}^\ell\to 0
\)
and
\(
\bm{\delta}_{r'r}^\ell\to 0
\; \forall (r,r') \in \mathcal{O}\), as $\ell \to \infty$.

The next theorem shows that, under standard convexity assumptions, the inner ADMM loop converges to the unique solution of the proximal state subproblem; see  Appendix~\ref{app:proof_proximal_state_admm_convergence} for the proof.

\begin{theorem}
\label{thm:proximal_state_admm_convergence}
Fix an outer iterate \(k\), and suppose that
\(
\{\mathbf{C}_r^k\}_{r=1}^R
\)
satisfies \eqref{eq:multi_region_estimation_problem_confusion}. Assume that, for every \(r\), the function
\(
D(\hat{\bm{\pi}}_r,\cdot)
\)
is closed, proper, and convex, and that the proximal state subproblem \eqref{eq:proximal_state_subproblem_overlap} has an optimal solution. Then, for any \(\beta>0\) and any initialization of the inner ADMM iterations \eqref{eq:admm_generic_updates}, the residuals satisfy
\(
\bm{\delta}_{rr'}^\ell\to 0
\)
and
\(
\bm{\delta}_{r'r}^\ell\to 0 \; \forall (r,r')\in\mathcal{O}\), the objective values converge to the optimal value of \eqref{eq:proximal_state_subproblem_overlap}, and the primal sequence
\(
(\{\bm{\rho}_r^\ell\},\{\bm{\rho}_{rr'}^\ell\})
\)
converges to the unique solution of \eqref{eq:proximal_state_subproblem_overlap}.
\end{theorem}

\begin{remark}
Theorem~\ref{thm:proximal_state_admm_convergence} concerns only the inner ADMM loop for the state-update step at a fixed outer iterate \(k\). It does not assert convergence of the full outer alternating scheme in which the confusion matrices are subsequently updated by \eqref{eq:admm_confusion_update}.
\end{remark}

\subsection{Communication and Computation Complexities}

We next quantify the communication and memory scaling of our novel regional formulation. Specifically, we compare the number of free parameters in the global and regional models and bound the communication incurred by one inner ADMM iteration. Let
\(
P_{\mathrm{glob}}:=(4^q-1)+M(M-1)
\)
and
\(
P_{\mathrm{reg}}:=\sum_{r=1}^R\bigl[(4^{q_r}-1)+M_r(M_r-1)\bigr]
\)
denote the numbers of free real parameters in the global and regional models, respectively. Define the memory-compression factor
\(
F_{\mathrm{mem}}:=P_{\mathrm{glob}}/P_{\mathrm{reg}}
\),
and let
\(
q_{\max}:=\max_r q_r
\),
\(
q_{\min}:=\min_r q_r
\),
\(
q_{\mathrm{ov},\max}:=\max_{(r,r')\in\mathcal O} q_{rr'}
\),
and
\(
d_{\max}:=\max_r |\{r':(r,r')\in\mathcal O\}|
\),
where \(d_{\max}\) is the maximum degree of the overlap graph. Assume further that the regional POVM sizes satisfy
\(
4^{q_r}\le M_r\le \mu\,4^{q_r}
\)
for all \(r\) and some constant \(\mu\ge 1\). Finally, counting both transmission directions in one inner ADMM iteration, define
\(
N_{\mathrm{comm}}:=\sum_{r=1}^R\sum_{r':(r,r')\in\mathcal O}4^{q_{rr'}}.
\)

\begin{proposition}
\label{prop:quantitative_scaling_advantage}
If \(M\ge 4^q\) and \(4^{q_r}\le M_r\le \mu\,4^{q_r}\) $\forall$ \(r\), then
\begin{equation}
\label{eq:quantitative_scaling_advantage_bounds}
\begin{aligned}
P_{\mathrm{reg}}
&\le (1+\mu^2)\,R\,16^{q_{\max}},\\
P_{\mathrm{glob}}
&\ge 16^q-4^q,\\
F_{\mathrm{mem}}
&\ge \frac{16^q-4^q}{(1+\mu^2)\,R\,16^{q_{\max}}},\\
N_{\mathrm{comm}}
&\le d_{\max}R\,4^{q_{\mathrm{ov},\max}},\\
\frac{N_{\mathrm{comm}}}{P_{\mathrm{reg}}}
&\le
\frac{d_{\max}}{1-4^{-q_{\min}}}\,
4^{\,q_{\mathrm{ov},\max}-2q_{\min}}.
\end{aligned}
\end{equation}
In particular, if \(q_{\max}\), \(q_{\mathrm{ov},\max}\), and \(d_{\max}\) remain uniformly bounded as \(q\) grows, then
\(
P_{\mathrm{reg}}=O(R),
\)
\(
N_{\mathrm{comm}}=O(R),
\)
and
\(
P_{\mathrm{glob}}=O(16^q)
\);
moreover, the lower bound on \(F_{\mathrm{mem}}\) in \eqref{eq:quantitative_scaling_advantage_bounds} grows exponentially in \(q-q_{\max}\).
\end{proposition}

Proposition~\ref{prop:quantitative_scaling_advantage} makes the scaling separation explicit: the memory cost of the regional model depends exponentially on the largest region size \(q_{\max}\), whereas the per-inner-iteration communication depends exponentially only on the largest overlap size \(q_{\mathrm{ov},\max}\), not on the full-system size \(q\). The proof is given in Appendix~\ref{app:proof_quantitative_scaling_advantage}.


\section{Simulations}
\label{sec:simulations}

\subsection{Simulation Setup}
\label{subsec:simulation_setup}

As no standard public benchmark combines overlapping regional tomography with regional readout-confusion matrices, we use a synthetic benchmark constructed for~\eqref{eq:multi_region_estimation_problem}. It leverages established tomography workflows and public resources, including QDataSet, Qiskit Experiments, and pyGSTi~\cite{Perrier2022QDataSet,Perrier2021QDataSetRepo,Kanazawa2023QiskitExperiments,QiskitTomographyDocs2025,Nielsen2020pyGSTi,Nielsen2021GST}. Unless stated otherwise, all experiments use the default settings in Table~\ref{tab:simulation_setup}.

For each experiment, we generate the global \(q\)-qubit ground-truth state as
\begin{equation}
\label{eq:sim_setup_global_state}
\bm{\rho}^{\star}
=
(1-\nu)\ket{\bm{\psi}}\!\bra{\bm{\psi}}
+
\frac{\nu}{Q} {\mathbf I}
\qquad
\ket{\bm{\psi}} := \mathbf U \ket{0}^{\otimes q},
\end{equation}
where \(\mathbf U\in\mathbb{C}^{Q\times Q}\) is a random unitary matrix and
\(\nu\in[0,1)\) sets the mixing weight between the pure state
\(\ket{\bm{\psi}}\!\bra{\bm{\psi}}\) and the maximally mixed state; that is $(1/Q)\mathbf I$. For region \(r\), the ground-truth state \(\bm{\rho}_r^\star\) is
obtained by reducing \(\bm{\rho}^{\star}\) to region \(r\).

Per region \(r\), we use
\(
M_r = Q_r^2
\)
outcomes, where
\(
Q_r = 2^{q_r}.
\)
The ground-truth \(\mathbf C_r^\star\) is generated by perturbing the ideal confusion matrix \(\mathbf I_r\) and projecting the result onto the feasible set \(\mathcal C_r\). To quantify the \emph{deviation from ideal readout}, we define
\begin{equation}
\label{eq:sim_setup_deltaC}
\delta_C^\star
:=
\frac{1}{R}
\sum_{r=1}^R
\frac{\|\mathbf C_r^\star-\mathbf I_r\|_F}{\|\mathbf I_r\|_F}.
\end{equation}
Thus, \(\delta_C^\star\) is the average relative distance between the ground-truth confusion matrices and the ideal confusion matrix.

In all simulations, we set \(q_s=1\) per site. We consider four regional graph geometries, namely \textsc{Ring}, \textsc{Ladder}, \textsc{Torus}, and \textsc{Hub}, as illustrated in Fig.~\ref{fig:regional_geometries}. In all cases, each region \(\mathcal{R}_r\) contains four sites (\(|\mathcal{R}_r|=4\), so \(q_r=4\)), and each overlap contains two sites (\(q_{\mathrm{ov}}=2\)). The \textsc{Ring} uses length-\(4\) regions on a periodic cycle; the \textsc{Ladder} uses overlapping \(2\times2\) plaquettes on a periodic two-leg ladder; the \textsc{Torus} uses overlapping \(2\times2\) plaquettes on a periodic square lattice; and the \textsc{Hub} uses \(4\)-site regions sharing a common \(2\)-site core. The total number of qubits \(q\) and the number of regions \(R\) for each geometry are given in Fig.~\ref{fig:regional_geometries}.

For each region \(r\), we use \(T_r\) measurement shots to form~\(\hat{\bm{\pi}}_r\). The total number of shots is
\begin{equation*}
T_{\mathrm{tot}} := \sum_{r=1}^R T_r.
\end{equation*}

We compare three estimators. The fixed-ideal estimator \(\mathsf I\) assumes \(\mathbf C_r=\mathbf I_r, \forall r\). The proposed joint estimator \(\mathsf J\) jointly estimates \(\{\bm{\rho}_r,\mathbf C_r\}_{r=1}^R\) by solving \eqref{eq:multi_region_estimation_problem}. Finally, the oracle denoted by \(\mathsf O\) uses the true confusion matrices \(\mathbf C_r^\star\) to estimate the states. Thus, the \(\mathsf I\) estimator ignores the readout errors, the \(\mathsf J\) estimator estimates them, and \(\mathsf O\) gives the oracle benchmark with perfect readout knowledge. 

We use the \emph{state} and \emph{confusion-matrix} errors as the main performance metrics
\begin{equation}
\label{eq:sim_setup_erho}
e_{\rho}
=
\frac{1}{R}\sum_{r=1}^R
\frac{\|\hat{\bm{\rho}}_r-\bm{\rho}_r^\star\|_F}{\|\bm{\rho}_r^\star\|_F},
\quad
e_C
=
\frac{1}{R}\sum_{r=1}^R
\frac{\|\hat{\mathbf C}_r-\mathbf C_r^\star\|_F}{\|\mathbf C_r^\star\|_F}.
\end{equation}

\begin{table}[t]
\centering
\caption{Default simulation settings used unless stated otherwise.}
\label{tab:simulation_setup}
\setlength{\tabcolsep}{5pt}
\begin{tabular}{lll}
\toprule
Category & Symbol & Default choice \\
\midrule
Site size & \(q_s\) & \(1\) qubit per site \\
Region size & \(q_r\) & \(4\) qubits per region \\
Overlap size & \(q_{\mathrm{ov}}\) & \(2\) qubits \\
$\#$ of measurements & \(M_r\) & \(Q_r^2 = 4^{q_r} = 256\) \\
Shots per region & \(T_r\) & \(10^4\) unless \(T_{\mathrm{tot}}\) is varied \\
Total shots & \(T_{\mathrm{tot}}\) & \(\sum_{r=1}^R T_r\) \\
Discrepancy & \(D(\mathbf a,\mathbf b)\) & \(\tfrac12\|\mathbf a-\mathbf b\|_2^2\) \\
Readout regularizer & \(\Omega_r(\mathbf C_r)\) & \(\|\mathbf C_r-\mathbf I\|_F^2\) \\
ADMM penalty & \(\beta\) & \(1\) \\
State proximal weight & \(\gamma_{\rho}\) & \(0.1\) \\
Confusion proximal weight & \(\gamma_C\) & \(0.1\) \\
Readout regularization & \(\lambda\) & \(10^{-2}\) unless varied \\
Quantum system state & \(\bm{\rho}^\star\) & mixed state in \eqref{eq:sim_setup_global_state} \\
Regional confusion matrix & \(\mathbf C_r^\star\) & perturb \(\mathbf I_r\), project onto \(\mathcal C_r\) \\
Graph geometries & -- & \textsc{Ring}, \textsc{Ladder}, \textsc{Torus}, \textsc{Hub} \\
Estimators & -- & ideal \(\mathsf I\), joint \(\mathsf J\), oracle \(\mathsf O\) \\
Performance measures & -- & \(e_{\rho}\), \(e_C\) \\
\bottomrule
$\star$ stands for ground truth.
\end{tabular}
\end{table}

\begin{figure}[t]
\centering

\subfloat[Ring (\(q=12\), \(R=6\))\label{fig:geom_ring}]{%
  \begin{minipage}[t]{0.48\columnwidth}
    \centering
    \includegraphics[width=\linewidth]{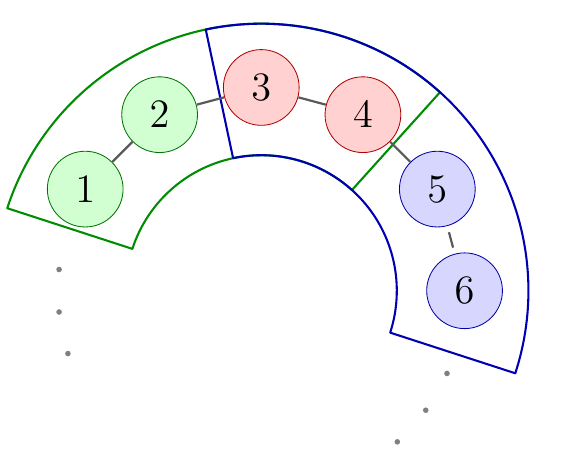}
  \end{minipage}
}
\hfill
\subfloat[Ladder (\(q=12\), \(R=6\))\label{fig:geom_ladder}]{%
  \begin{minipage}[t]{0.48\columnwidth}
    \centering
    \includegraphics[width=\linewidth]{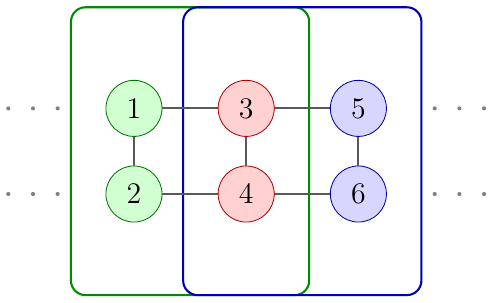}
  \end{minipage}
}

\vspace{0.4em}

\subfloat[Torus (\(q=16\), \(R=9\))\label{fig:geom_torus}]{%
  \begin{minipage}[t]{0.48\columnwidth}
    \centering
    \includegraphics[width=\linewidth]{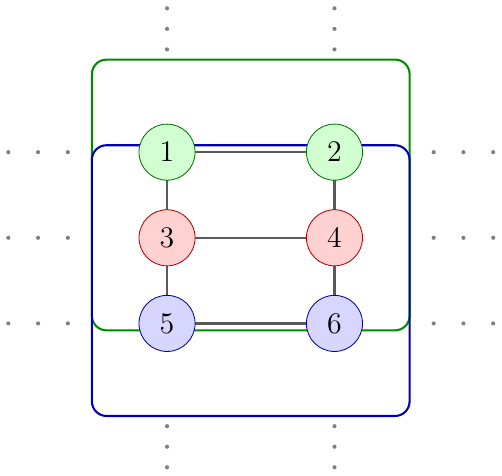}
  \end{minipage}
}
\hfill
\subfloat[Hub (\(q=14\), \(R=6\))\label{fig:geom_hub}]{%
  \begin{minipage}[t]{0.48\columnwidth}
    \centering
    \includegraphics[width=\linewidth]{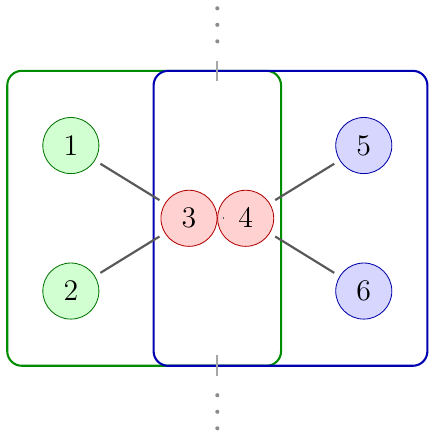}
  \end{minipage}
}

\caption{
Regional graph geometries used in the simulations. Each sub-figure shows a representative overlapping pair of regions. Red nodes denote the shared sites, while green and blue nodes belong only to the left and right region, respectively. In the \textsc{Hub} geometry, the shared sites form the common core.
}
\label{fig:regional_geometries}
\end{figure}

\subsection{Simulation Tests}

\begin{figure}[t]
\centering
\subfloat[Consensus residual error\label{fig:admm_dynamics_topology_a}]{%
\includegraphics[width=\columnwidth]{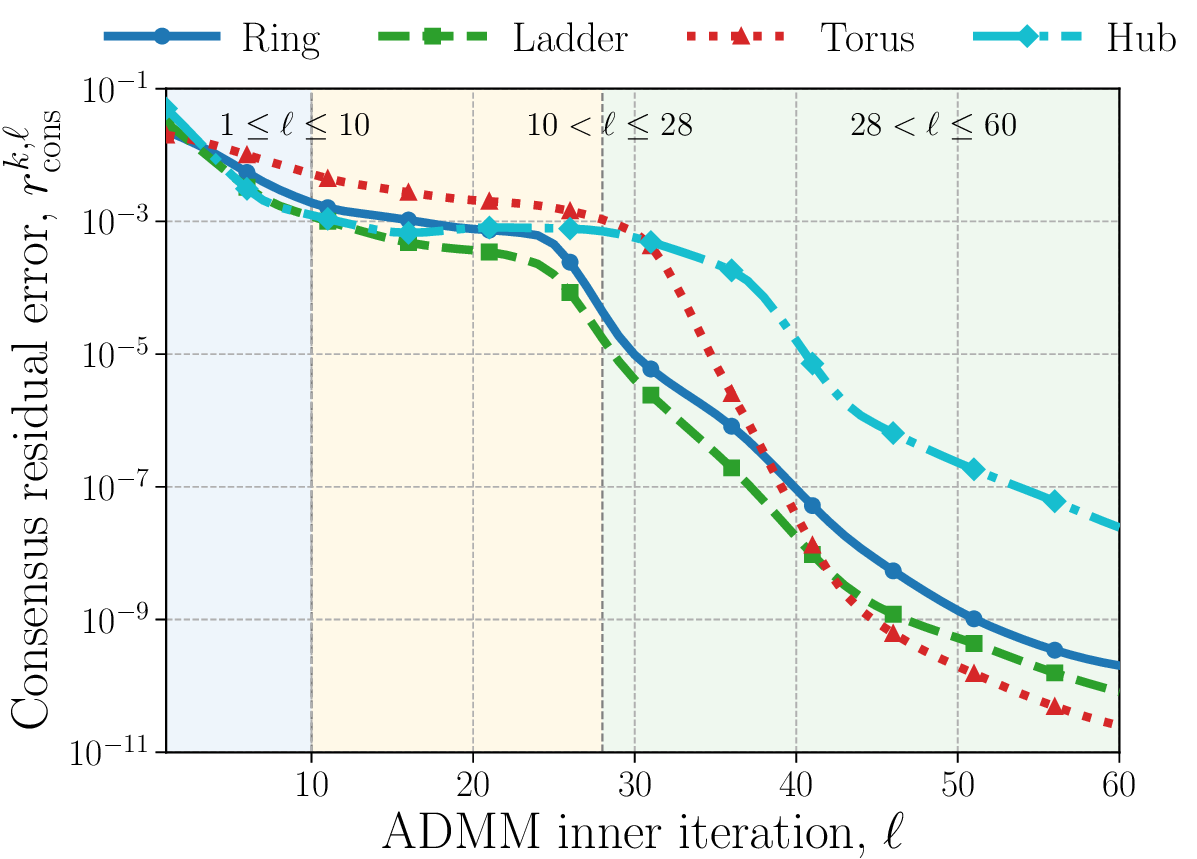}
}
\vspace{0.35em}
\subfloat[state optimality gap\label{fig:admm_dynamics_topology_b}]{%
\includegraphics[width=\columnwidth]{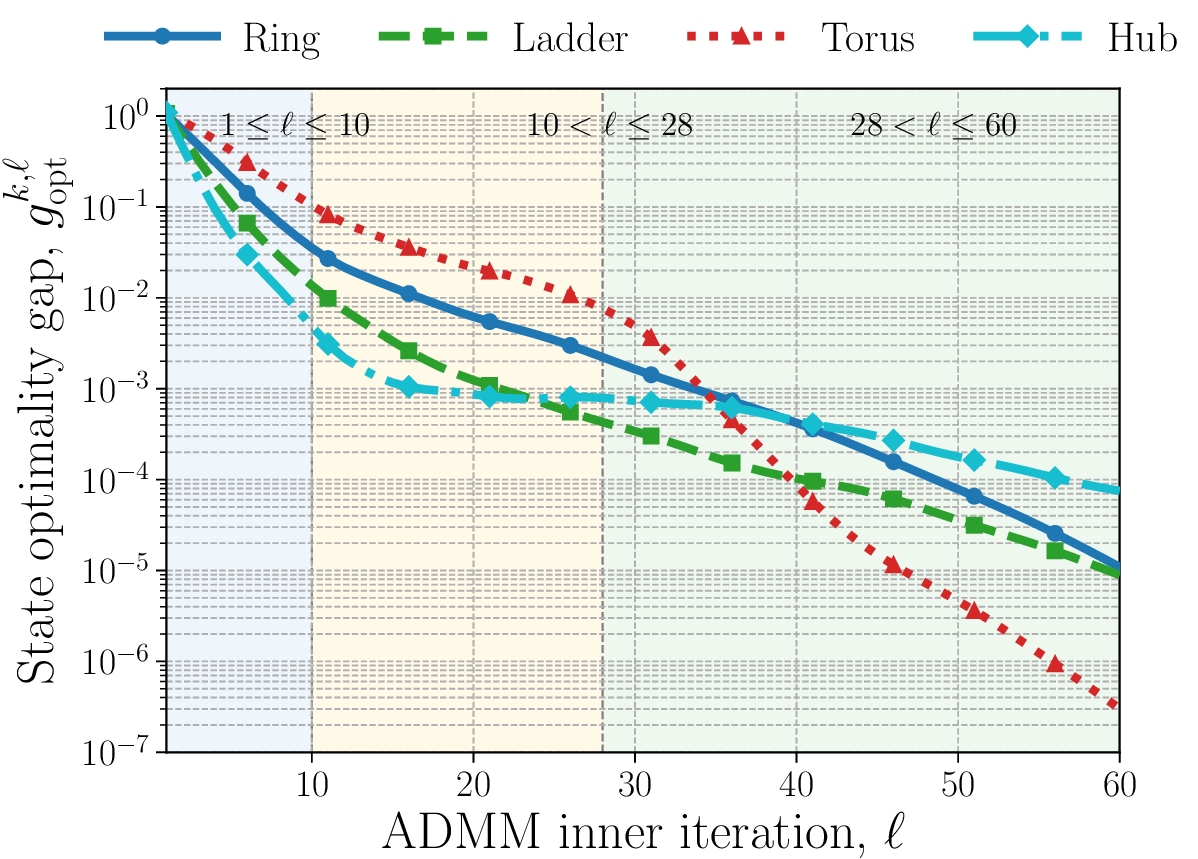}
}
\caption{Inner ADMM convergence for the four graph geometries.}
\label{fig:admm_dynamics_topology}
\end{figure}

To study how the graph geometry affects the inner ADMM step, we define the \emph{consensus residual} error
\begin{equation*}
r_{\mathrm{cons}}^{k,\ell}
\! :=
\Biggl(
\sum_{(r,r')\in\mathcal O} \!\!\! 
\|\bm{\rho}^{k,\ell}_{r}[r']-\bm{\rho}_{rr'}^{k,\ell}\|_F^2
+
\|\bm{\rho}^{k,\ell}_{r'}[r]-\bm{\rho}_{r'r}^{k,\ell}\|_F^2
\Biggr)^{1/2}\!\!\!\!\!
\end{equation*}
and the normalized \emph{state optimality} gap
\begin{equation*}
g_{\mathrm{opt}}^{k,\ell}
:=
\frac{
\mathcal J_{\rho}^{k}\!\bigl(\{\bm{\rho}_r^{k,\ell}\}_{r=1}^R\bigr)
-
\mathcal J_{\rho,\min}^{k}
}{
\max\{1,|\mathcal J_{\rho,\min}^{k}|\}
}
\end{equation*}
where
\begin{equation*}
\begin{aligned}
\mathcal J_{\rho}^{k}\!\bigl(\{\bm{\rho}_r\}; \{\mathbf C_r\}\bigr)
:=
\sum_{r=1}^R \Bigl[
\frac12
\left\|
\hat{\bm{\pi}}_r-\mathbf C_r^k\bm{\pi}_r(\bm{\rho}_r)
\right\|_2^2
\\
{}+
\frac{\gamma_\rho}{2}
\|\bm{\rho}_r-\bm{\rho}_r^k\|_F^2
\Bigr],
\\[0.3ex]
\MoveEqLeft[15]
\mathcal J_{\rho,\min}^{k}
:=
\min_{\{\bm{\rho}_r\}}
\mathcal J_{\rho}^{k}\!\bigl(\{\bm{\rho}_r\}_{r=1}^R\bigr).
\end{aligned}
\end{equation*}

Fig.~\ref{fig:admm_dynamics_topology_a} plots \(r_{\mathrm{cons}}^{k,\ell}\), and Fig.~\ref{fig:admm_dynamics_topology_b} plots \(g_{\mathrm{opt}}^{k,\ell}\) versus inner ADMM iteration \(\ell\). Fig.~\ref{fig:admm_dynamics_topology} demonstrates three main stages: early with \(1 \leq \ell \leq 10\), transition with \(10 < \ell \leq 28\), and late with \(28 < \ell \leq 60\). The \textsc{Hub} drops fastest in the first stage because all regions share the same core. Later, that same core becomes a bottleneck, hindering further reduction of the state optimality gap. The \textsc{Ring} improves more slowly because agreement propagates only through local nearest-neighbor overlaps. The \textsc{Ladder} is more balanced: it adds overlap constraints without creating a single bottleneck. The \textsc{Torus} exhibits the strongest late-stage decay because its denser overlap pattern offers rapid consensus among sites, while improving \(g_{\mathrm{opt}}^{k,\ell}\).

\begin{figure}[t]
\centering
\subfloat[State error\label{fig:outer_dynamics_regularization_a}]{%
\includegraphics[width=\columnwidth]{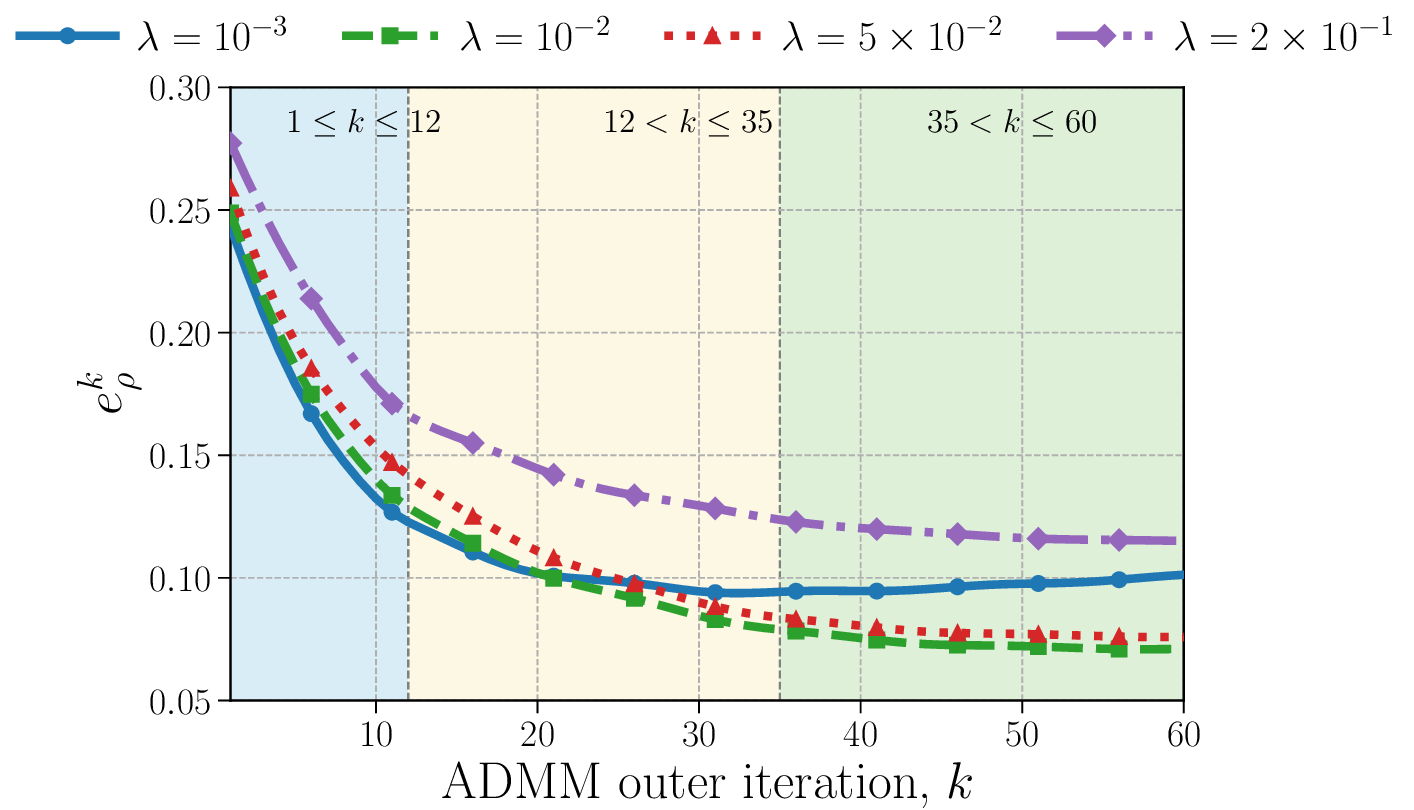}
}
\vspace{0.35em}
\subfloat[Confusion-matrix error\label{fig:outer_dynamics_regularization_b}]{%
\includegraphics[width=\columnwidth]{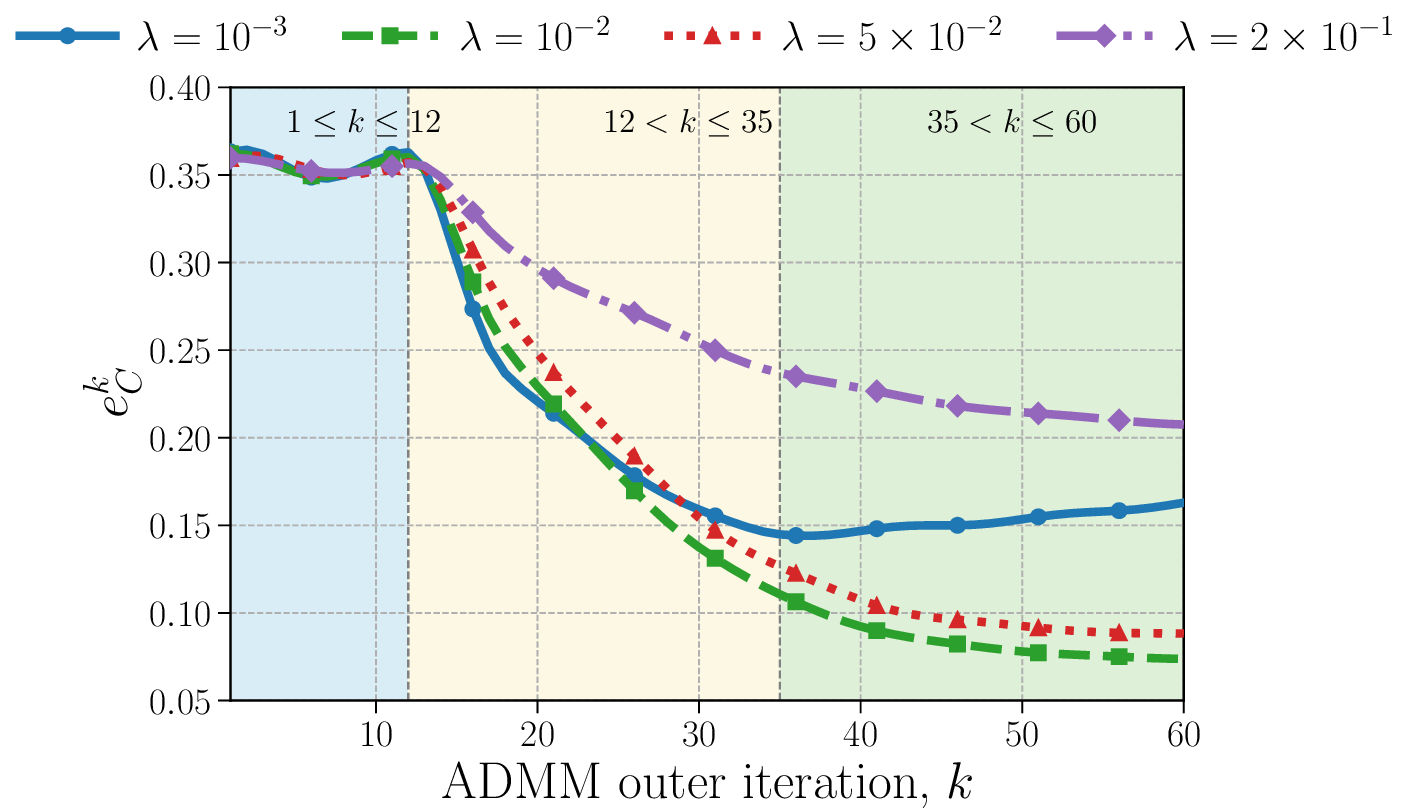}
}
\caption{Effect of the readout regularization parameter \(\lambda\) on the outer alternating updates.}
\label{fig:outer_dynamics_regularization}
\end{figure}

To study the effect of the readout regularization hyperparameter \(\lambda\) on the outer alternating updates, we track the state error \(e_\rho^k\) and the confusion-matrix error \(e_C^k\) across the outer iterations \(k\). These are defined as in \eqref{eq:sim_setup_erho}, with \(\hat{\bm{\rho}}_r\) and \(\hat{\mathbf C}_r\) replaced by \(\bm{\rho}_r^k\) and \(\mathbf C_r^k\), respectively. In this experiment, we use the \textsc{Ladder} geometry for various
\(
\lambda\in\{10^{-3},10^{-2},5\times 10^{-2},2\times 10^{-1}\}.
\)

Figures~\ref{fig:outer_dynamics_regularization_a} and~\ref{fig:outer_dynamics_regularization_b} depict the influence of \(\lambda\) on error terms. Its main role is to judiciously control the updates of the confusion matrices over iterations. When \(\lambda\) is too small, the state error drops quickly at first, but the readout updates become less stable later. When \(\lambda\) is too large, the confusion matrices stay too close to \(\mathbf I_r\), so both errors decrease more slowly. The best balance appears at intermediate values. Here this occurs at \(\lambda=10^{-2}\), which gives the lowest final \(e_{\boldsymbol{\rho}}\) and~\(e_{\mathbf{C}}\).

\begin{figure}[t]
\centering
\subfloat[State-recovery gain\label{fig:phase_diagram_joint_estimation_a}]{%
\includegraphics[width=\columnwidth]{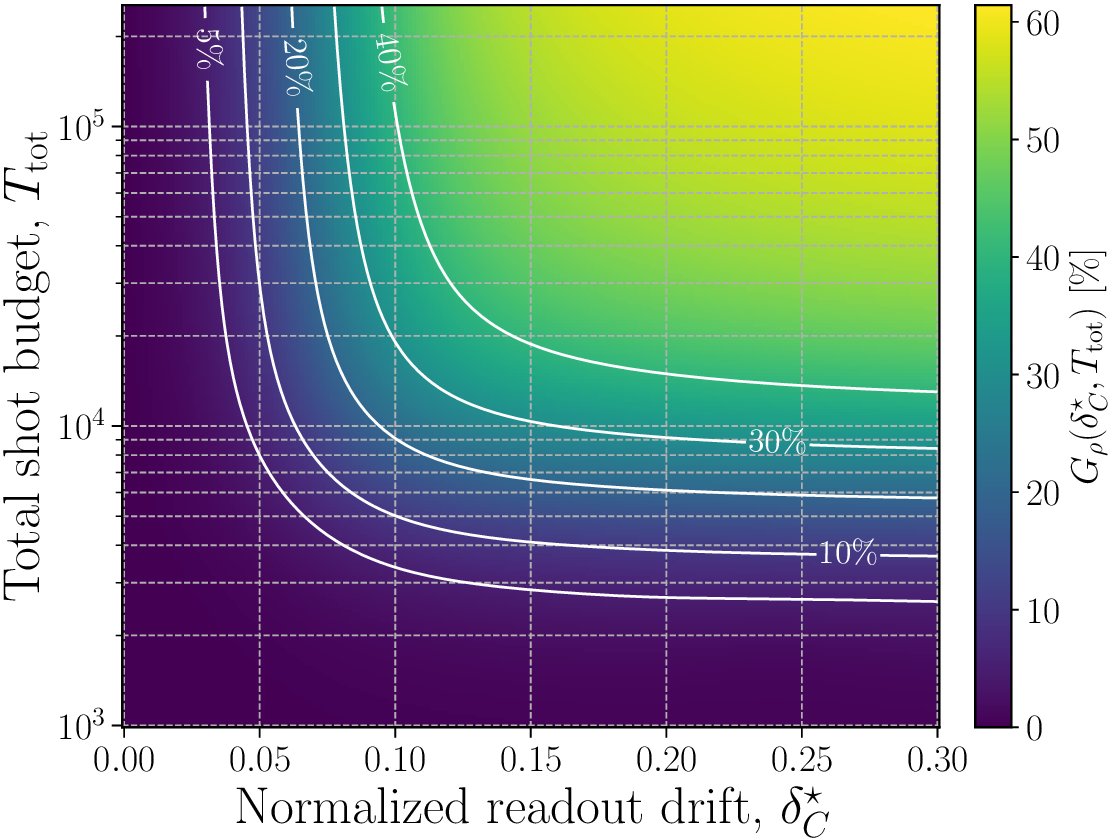}
}
\vspace{0.35em}
\subfloat[Oracle-gap \label{fig:phase_diagram_joint_estimation_b}]{%
\includegraphics[width=\columnwidth]{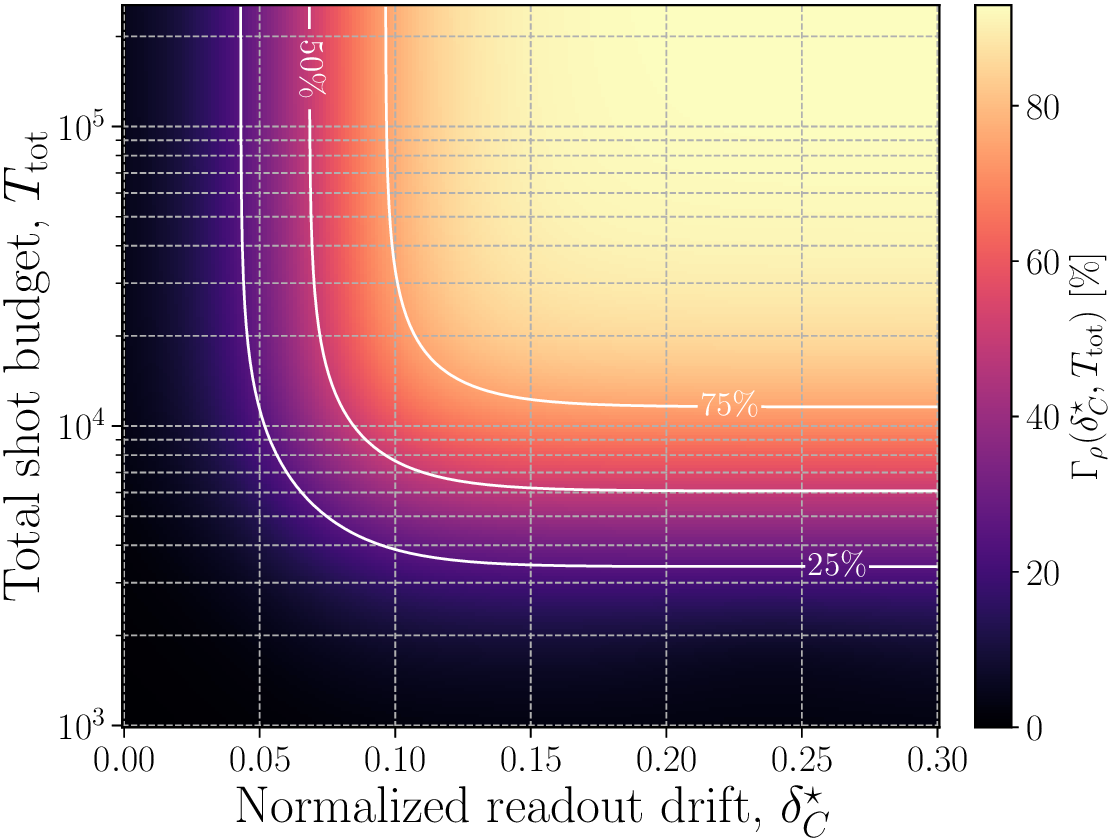}
}
\caption{Benefit of joint readout estimation as a function of \(\delta_C^\star\) and \(T_{\mathrm{tot}}\).}
\label{fig:phase_diagram_joint_estimation}
\end{figure}

The next set of experiments further investigate the pros and cons of our proposed joint state- and confusion-matrix estimation framework. We first use the \textsc{Ladder} geometry to compare the state errors of \(\mathsf I\), \(\mathsf J\), and \(\mathsf O\) estimators as functions of the deviation from ideal readout \(\delta_C^\star\) and the total number of shots \(T_{\mathrm{tot}}\). Let us define the \emph{state-recovery} gain of \(\mathsf J\) over \(\mathsf I\) as
\begin{equation*}
G_{\rho}(\delta_C^\star,T_{\mathrm{tot}})
:=
100\,
\frac{
e_{\rho}^{\mathsf I}(\delta_C^\star,T_{\mathrm{tot}})
-
e_{\rho}^{\mathsf J}(\delta_C^\star,T_{\mathrm{tot}})
}{
e_{\rho}^{\mathsf I}(\delta_C^\star,T_{\mathrm{tot}})
}.
\end{equation*}
This measures the percentage reduction in state error obtained by \(\mathsf J\) relative to the fixed-ideal baseline \(\mathsf I\). Also define the \emph{oracle-gap} as
\begin{equation*}
\Gamma_{\rho}(\delta_C^\star,T_{\mathrm{tot}})
:=
100\,
\frac{
e_{\rho}^{\mathsf I}(\delta_C^\star,T_{\mathrm{tot}})
-
e_{\rho}^{\mathsf J}(\delta_C^\star,T_{\mathrm{tot}})
}{
e_{\rho}^{\mathsf I}(\delta_C^\star,T_{\mathrm{tot}})
-
e_{\rho}^{\mathsf O}(\delta_C^\star,T_{\mathrm{tot}})
}.
\end{equation*}
Here the denominator is the full improvement from \(\mathsf I\) to the oracle \(\mathsf O\), so \(\Gamma_\rho=100\) means that \(\mathsf J\) matches the oracle, while \(\Gamma_\rho=0\) indicates that it gives no improvement over \(\mathsf I\). 

Fig.~\ref{fig:phase_diagram_joint_estimation_a} plots \(G_\rho\), and Fig.~\ref{fig:phase_diagram_joint_estimation_b} plots \(\Gamma_\rho\). Both gains are small when \(\delta_C^\star\) is close to zero, since $\mathbf{C}_r \approx \mathbf{I}_r$ (c.f., \eqref{eq:sim_setup_deltaC}). Gains are also marginal when \(T_{\mathrm{tot}}\) is small, as there are not enough measurements to accurately estimate the confusion matrices \(\mathbf{C}_r\). The sizable gains appear when both \(\delta_C^\star\) and \(T_{\mathrm{tot}}\) are large. In that regime, readout error is strong enough to matter, and the data are sufficient for \(\mathsf J\) to learn it well.

\begin{figure}[t]
\centering
\subfloat[Communication tradeoff\label{fig:pareto_tradeoff_a}]{%
\includegraphics[width=\columnwidth]{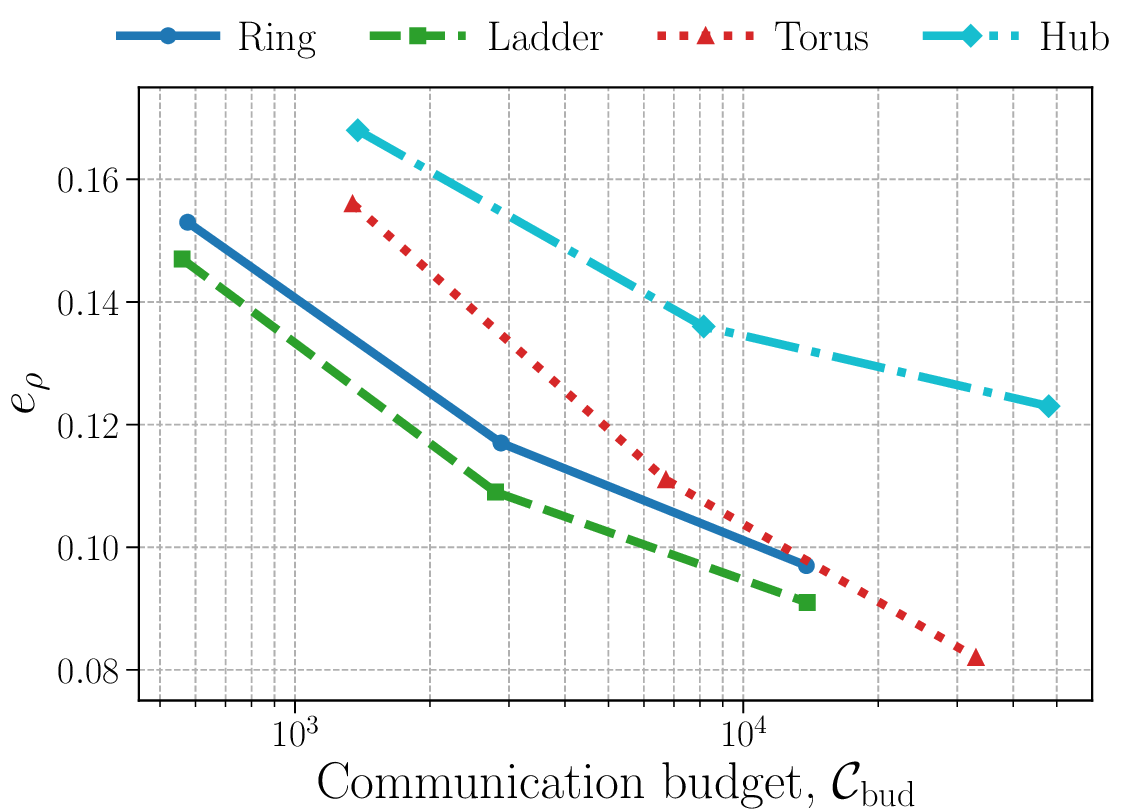}
}
\vspace{0.35em}
\subfloat[Computation tradeoff\label{fig:pareto_tradeoff_b}]{%
\includegraphics[width=\columnwidth]{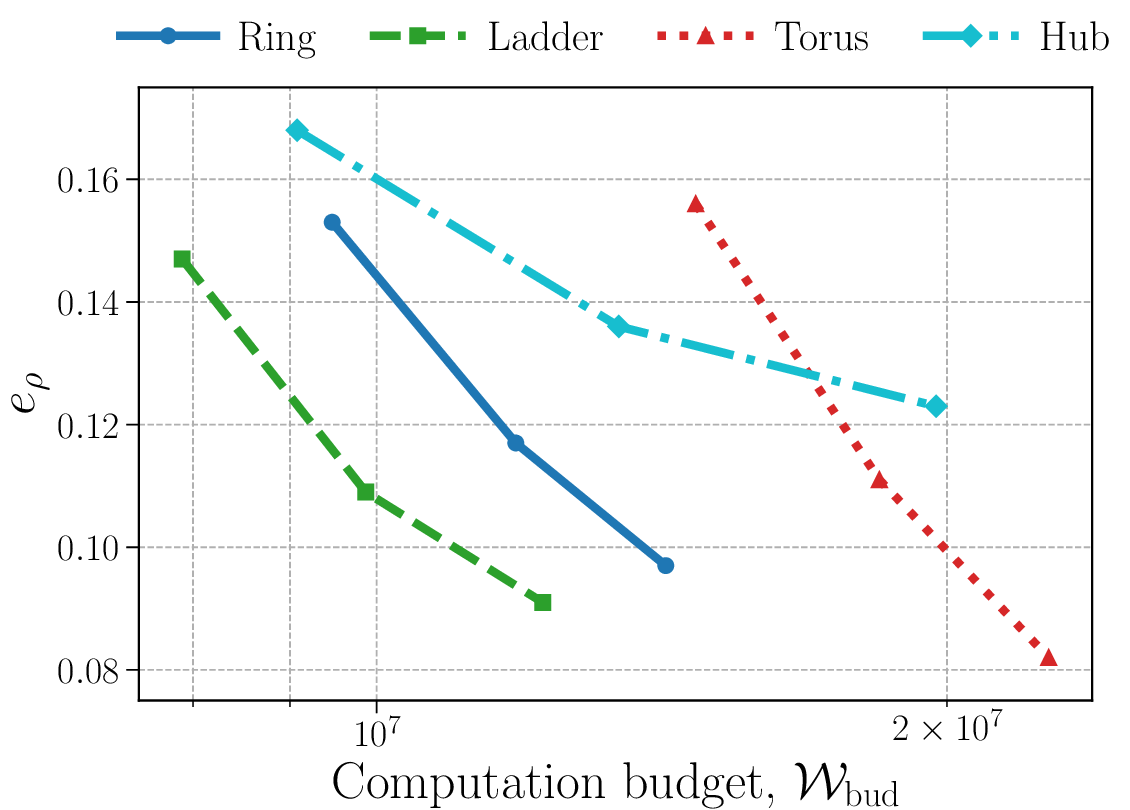}
}
\caption{Recovery-cost tradeoff for the four graph geometries.}
\label{fig:pareto_tradeoff}
\end{figure}

Fig.~\ref{fig:pareto_tradeoff} demonstrates the tradeoff between state recovery and algorithmic cost across the four considered graph geometries. The vertical axis is the state error \(e_\rho\), and the horizontal one is the communication cost defined as
\begin{equation*}
\mathcal C_{\mathrm{bud}}
:=
\bar L
\sum_{(r,r')\in\mathcal O}
4^{q_{rr'}},
\end{equation*}
and the computation cost
\begin{equation*}
\mathcal W_{\mathrm{bud}}
:=
\bar L
\left(
\sum_{r=1}^R 4^{q_r}
+
\sum_{r=1}^R M_r^2
\right),
\end{equation*}
where \(q_{rr'}\) is the number of qubits in the overlap between regions \(r\) and \(r'\), and \(\bar L\) is the average number of inner ADMM iterations needed to reach a fixed stopping tolerance. Fig.~\ref{fig:pareto_tradeoff_a} plots \(e_\rho\) versus \(\mathcal C_{\mathrm{bud}}\), and Fig.~\ref{fig:pareto_tradeoff_b} plots \(e_\rho\) versus \(\mathcal W_{\mathrm{bud}}\).

While we demonstrated earlier that \textsc{Torus} achieves the lowest state error, Fig.~\ref{fig:pareto_tradeoff_b} reveals that this happens at the largest cost, especially in computation. The \textsc{Ring} is communication and computation efficient, but its sparse overlap pattern limits recovery. The \textsc{Hub} is not attractive here, as its shared core increases coordination cost without giving the best recovery. The \textsc{Ladder} gives the best overall balance. It improves over the \textsc{Ring}, avoids the communication burden of the \textsc{Torus}, and avoids the bottleneck of the \textsc{Hub}.

\begin{table*}[t]
\centering
\caption{Benchmark summary for the four graph geometries. Lower is better for \(e_\rho\), \(e_C\), \(\mathcal C_{\mathrm{bud}}\), and \(\mathcal W_{\mathrm{bud}}\). Higher is better for \(G_\rho\) and \(\Gamma_\rho\).}
\label{tab:benchmark_summary}
\setlength{\tabcolsep}{4.2pt}
\begin{tabular}{lccccccccccc}
\toprule
\multirow{2}{*}{Geometry}
& \multicolumn{3}{c}{Structure}
& \multicolumn{3}{c}{State error}
& \multicolumn{1}{c}{Confusion}
& \multicolumn{2}{c}{Joint-estimation benefit}
& \multicolumn{2}{c}{Cost} \\
\cmidrule(lr){2-4}
\cmidrule(lr){5-7}
\cmidrule(lr){8-8}
\cmidrule(lr){9-10}
\cmidrule(lr){11-12}
& \(q\) & \(R\) & \(q_{\mathrm{ov}}\)
& \(e_\rho^{\mathsf I}\) & \(e_\rho^{\mathsf J}\) & \(e_\rho^{\mathsf O}\)
& \(e_C^{\mathsf J}\)
& \(G_\rho\,[\%]\) & \(\Gamma_\rho\,[\%]\)
& \(\mathcal C_{\mathrm{bud}}\) & \(\mathcal W_{\mathrm{bud}}\) \\
\midrule
Ring
& \(12\) & \(6\) & \(2\)
& \(0.146\) & \(0.117\) & \(0.094\)
& \(0.086\)
& \(19.9\) & \(55.8\)
& \(2.88\times 10^3\) & \(1.18\times 10^7\) \\

Ladder
& \(12\) & \(6\) & \(2\)
& \(\mathbf{0.141}\) & \(\mathbf{0.109}\) & \(0.087\)
& \(\mathbf{0.074}\)
& \(22.7\) & \(\mathbf{59.3}\)
& \(\mathbf{2.80\times 10^3}\) & \(\mathbf{9.87\times 10^6}\) \\

Torus
& \(16\) & \(9\) & \(2\)
& \(0.151\) & \(0.111\) & \(\mathbf{0.081}\)
& \(0.079\)
& \(\mathbf{26.5}\) & \(57.1\)
& \(6.72\times 10^3\) & \(1.84\times 10^7\) \\

Hub
& \(14\) & \(6\) & \(2\)
& \(0.161\) & \(0.136\) & \(0.101\)
& \(0.118\)
& \(15.5\) & \(41.7\)
& \(8.16\times 10^3\) & \(1.34\times 10^7\) \\
\bottomrule
\end{tabular}
\end{table*}

Finally, Table~\ref{tab:benchmark_summary} summarizes the main results at \(q_{\mathrm{ov}}=2\). The same pattern as in the figures appears here. The \(\mathsf J\) estimator improves over \(\mathsf I\) for all four geometries. The \textsc{Ladder} exhibits the best overall balance, where the \(\mathsf J\) estimator has the lowest state and confusion-matrix errors, the largest oracle gap, and the lowest communication and computation costs. The \textsc{Torus} gives the best oracle state error and the largest state-recovery gain, but at a much higher cost. The \textsc{Ring} remains a stable baseline, while the \textsc{Hub} is the weakest overall because its shared core increases communication and computation costs, and slows late-stage improvement.

\section{Conclusion}
\label{sec:conclusion}
This paper developed a framework for joint regional quantum state and readout estimation in multiqubit systems with overlapping structure. Each region is assigned a local density operator and a local confusion matrix, while neighboring regions are coupled through reduced-state consistency on shared subsystems. This leads to a structured bilinear optimization task, and a distributed proximal alternating ADMM method. The analytical guarantees established include local identifiability, local quadratic growth of the population misfit, and convergence of the inner state-update step. Simulations demonstrated that joint estimation improves state recovery over the ideal baseline, recovers a substantial portion of oracle performance, and reveals a performance-cost tradeoff across various graph topologies. 


\appendices

\section{Proof of Theorem~\ref{thm:local_identifiability}}
\label{app:proof_local_identifiability}

\begin{proof}
Let
\(
\Delta\bm{\rho}_r:=\bm{\rho}_r-\bm{\rho}_r^\ast
\)
and
\(
\Delta\mathbf{C}_r:=\mathbf{C}_r-\mathbf{C}_r^\ast
\)
for all \(r\). Since both
\(
(\{\bm{\rho}_r\},\{\mathbf{C}_r\})
\)
and
\(
(\{\bm{\rho}_r^\ast\},\{\mathbf{C}_r^\ast\})
\)
satisfy
\eqref{eq:multi_region_estimation_problem_state}--\eqref{eq:multi_region_estimation_problem_overlap}, their difference satisfies the linearized feasibility conditions in \eqref{eq:linearized_feasible_set_ast}; hence
\[
(\{\Delta\bm{\rho}_r\},\{\Delta\mathbf{C}_r\})\in\mathcal{T}_\ast.
\]
By the assumption of Theorem~\ref{thm:local_identifiability}, for any pair \((\{\Delta \bm{\rho}_r\}, \{\Delta \mathbf{C}_r\}) \in \mathcal{T}_\ast\), we have:
\[
\mathcal{A}_\ast(\{\Delta\bm{\rho}_r\},\{\Delta\mathbf{C}_r\})=0
\ \Longrightarrow\
\Delta\bm{\rho}_r=0,\ \Delta\mathbf{C}_r=0,\ \forall r,
\]
which is a condition on the linearized map \(\mathcal{A}_\ast\) over the linearized feasible set \(\mathcal{T}_\ast\). It is not the conclusion
\[
\mathbf{C}_r\bm{\pi}_r(\bm{\rho}_r)
=
\mathbf{C}_r^\ast\bm{\pi}_r(\bm{\rho}_r^\ast),
\; \forall r
\;\Longrightarrow\;
\bm{\rho}_r=\bm{\rho}_r^\ast,\ \
\mathbf{C}_r=\mathbf{C}_r^\ast,\; \forall r,
\]
which concerns the nonlinear model. The proof below shows that the former implies the latter locally. In that order, define the ambient space
\[
\mathcal{V}
:=
\prod_{r=1}^R \mathrm{Herm}(\mathcal{H}_{q_r})
\times
\prod_{r=1}^R \mathbb{R}^{M_r\times M_r}.
\]
Then \(\dim(\mathcal{V})<\infty\). By \eqref{eq:linearized_feasible_set_ast}, \(\mathcal{T}_\ast\) is a linear subspace of \(\mathcal{V}\). Hence \(\mathcal{T}_\ast\) is finite-dimensional and closed. Define
\[
\mathcal{S}
:=
\left\{
(\{\Delta\bm{\rho}_r\},\{\Delta\mathbf{C}_r\})\in\mathcal{T}_\ast:
d_\ast=1
\right\}.
\]
Then \(\mathcal{S}\) is compact. Since
\[
(\{\Delta\bm{\rho}_r\},\{\Delta\mathbf{C}_r\})
\mapsto
\left\|
\mathcal{A}_\ast(\{\Delta\bm{\rho}_r\},\{\Delta\mathbf{C}_r\})
\right\|_2
\]
is continuous on \(\mathcal{S}\), it attains its minimum there. Let
\[
\alpha
:=
\min_{(\{\Delta\bm{\rho}_r\},\{\Delta\mathbf{C}_r\})\in\mathcal{S}}
\left\|
\mathcal{A}_\ast(\{\Delta\bm{\rho}_r\},\{\Delta\mathbf{C}_r\})
\right\|_2 .
\]
If \(\alpha=0\), then there exists
\(
(\{\Delta\bm{\rho}_r\},\{\Delta\mathbf{C}_r\})\in\mathcal{S}
\)
such that
\[
\mathcal{A}_\ast(\{\Delta\bm{\rho}_r\},\{\Delta\mathbf{C}_r\})=0,
\]
which implies
\(
\Delta\bm{\rho}_r=0
\)
and
\(
\Delta\mathbf{C}_r=0
\)
for all \(r\), contradicting \(d_\ast=1\). Hence \(\alpha>0\).

Now let
\(
(\{\Delta\bm{\rho}_r\},\{\Delta\mathbf{C}_r\})\in\mathcal{T}_\ast
\).
If \(d_\ast=0\), then \eqref{eq:appendix_linear_lower_bound_ast} is trivial. If \(d_\ast>0\), then
\[
\left(
\left\{\frac{\Delta\bm{\rho}_r}{d_\ast}\right\},
\left\{\frac{\Delta\mathbf{C}_r}{d_\ast}\right\}
\right)\in\mathcal{S},
\]
and therefore, by linearity of \(\mathcal{A}_\ast\),
\begin{equation}
\label{eq:appendix_linear_lower_bound_ast}
\begin{aligned}
\left\|
\mathcal{A}_\ast(\{\Delta\bm{\rho}_r\},\{\Delta\mathbf{C}_r\})
\right\|_2
&=
d_\ast
\left\|
\mathcal{A}_\ast\!\left(
\left\{\frac{\Delta\bm{\rho}_r}{d_\ast}\right\},
\left\{\frac{\Delta\mathbf{C}_r}{d_\ast}\right\}
\right)
\right\|_2 \\
&\ge
\alpha\, d_\ast .
\end{aligned}
\end{equation}
where \(d_\ast\) is defined in \eqref{eq:local_distance_ast}.

Next, using the linearity of \(\bm{\pi}_r(\cdot)\), we obtain for each \(r\),
\begin{equation}
\label{eq:appendix_exact_prediction_expansion_ast}
\begin{aligned}
\mathbf{C}_r\bm{\pi}_r(\bm{\rho}_r)
-\mathbf{C}_r^\ast\bm{\pi}_r(\bm{\rho}_r^\ast)
&=
\Delta\mathbf{C}_r\,\bm{\pi}_r(\bm{\rho}_r^\ast)
+
\mathbf{C}_r^\ast\,\bm{\pi}_r(\Delta\bm{\rho}_r)
\\
&\quad
+
\Delta\mathbf{C}_r\,\bm{\pi}_r(\Delta\bm{\rho}_r).
\end{aligned}
\end{equation}
For each \(r\), define
\[
\kappa_r
:=
\max_{\substack{X\in \mathrm{Herm}(\mathcal{H}_{q_r})\\ \|X\|_F=1}}
\|\bm{\pi}_r(X)\|_2 .
\]
Since \(\mathrm{Herm}(\mathcal{H}_{q_r})\) is finite-dimensional, its unit Frobenius sphere is compact, and
\(
X\mapsto \|\bm{\pi}_r(X)\|_2
\)
is continuous, one has \(\kappa_r<\infty\). Hence, for every Hermitian \(\Delta\bm{\rho}_r\),
\[
\|\bm{\pi}_r(\Delta\bm{\rho}_r)\|_2
=
\|\Delta\bm{\rho}_r\|_F
\left\|
\bm{\pi}_r\!\left(
\frac{\Delta\bm{\rho}_r}{\|\Delta\bm{\rho}_r\|_F}
\right)
\right\|_2
\le
\kappa_r\|\Delta\bm{\rho}_r\|_F
\]
whenever \(\Delta\bm{\rho}_r\neq 0\), and the same inequality is trivial when \(\Delta\bm{\rho}_r=0\). Let
\(
\kappa:=\max_r \kappa_r
\).
Then
\begin{align}
&
\left(
\sum_{r=1}^R
\|\Delta\mathbf{C}_r\,\bm{\pi}_r(\Delta\bm{\rho}_r)\|_2^2
\right)^{1/2}
\notag\\
&\le
\kappa
\left(
\sum_{r=1}^R
\|\Delta\mathbf{C}_r\|_F^2\|\Delta\bm{\rho}_r\|_F^2
\right)^{1/2}
\notag\\
&\le
\kappa
\left(
\sum_{r=1}^R \|\Delta\mathbf{C}_r\|_F^2
\right)^{1/2}
\left(
\sum_{r=1}^R \|\Delta\bm{\rho}_r\|_F^2
\right)^{1/2}
\notag\\
&\le
\frac{\kappa}{2}\, d_\ast^2 .
\label{eq:appendix_remainder_bound_ast}
\end{align}

Combining \eqref{eq:local_population_misfit_ast}, \eqref{eq:appendix_exact_prediction_expansion_ast}, \eqref{eq:appendix_linear_lower_bound_ast}, and \eqref{eq:appendix_remainder_bound_ast} gives
\begin{align}
\sqrt{2\,\mathcal{Q}_\ast\bigl(\{\bm{\rho}_r\},\{\mathbf{C}_r\}\bigr)}
&=
\left(
\sum_{r=1}^R
\left\|
\mathbf{C}_r\bm{\pi}_r(\bm{\rho}_r)
-
\mathbf{C}_r^\ast\bm{\pi}_r(\bm{\rho}_r^\ast)
\right\|_2^2
\right)^{1/2}
\notag\\
&\ge
\alpha\, d_\ast-\frac{\kappa}{2}\,d_\ast^2 .
\label{eq:appendix_master_lower_bound_ast}
\end{align}

Choose \(\varepsilon>0\) so that
\(
\varepsilon<\alpha/\kappa
\)
when \(\kappa>0\); if \(\kappa=0\), any \(\varepsilon>0\) will do. Then every feasible pair with \(d_\ast<\varepsilon\) satisfies
\[
\sqrt{2\,\mathcal{Q}_\ast\bigl(\{\bm{\rho}_r\},\{\mathbf{C}_r\}\bigr)}
\ge
\frac{\alpha}{2}\, d_\ast,
\]
and therefore
\[
\mathcal{Q}_\ast\bigl(\{\bm{\rho}_r\},\{\mathbf{C}_r\}\bigr)
\ge
\frac{\alpha^2}{8}\, d_\ast^2.
\]

Hence, the quadratic-growth claim in Theorem~1 holds with $c := \alpha^2/8$.

Finally, let \(d_\ast<\varepsilon\) and suppose that
\[
\mathbf{C}_r\bm{\pi}_r(\bm{\rho}_r)
=
\mathbf{C}_r^\ast\bm{\pi}_r(\bm{\rho}_r^\ast),
\qquad
\forall r.
\]
Then
\[
\mathcal{Q}_\ast\bigl(\{\bm{\rho}_r\},\{\mathbf{C}_r\}\bigr)=0.
\]
Hence
\[
0
=
\mathcal{Q}_\ast\bigl(\{\bm{\rho}_r\},\{\mathbf{C}_r\}\bigr)
\ge
c\, d_\ast^2 .
\]
Since \(c>0\), it follows that \(d_\ast=0\). By \eqref{eq:local_distance_ast},
\[
\bm{\rho}_r=\bm{\rho}_r^\ast,
\qquad
\mathbf{C}_r=\mathbf{C}_r^\ast,
\qquad
r=1,\dots,R.
\]
Thus the model is locally identifiable at
\(
(\{\bm{\rho}_r^\ast\},\{\mathbf{C}_r^\ast\})
\).
\end{proof}

\section{Proof of Theorem~\ref{thm:proximal_state_admm_convergence}}
\label{app:proof_proximal_state_admm_convergence}

\begin{proof}
Fix an outer iterate \(k\), and regard \(\{\mathbf{C}_r^k\}_{r=1}^R\) as fixed. For each \(r\), the map
$
\bm{\rho}_r \mapsto \bm{\pi}_r(\bm{\rho}_r)
$
is linear; hence so is
$
\bm{\rho}_r \mapsto \mathbf{C}_r^k \bm{\pi}_r(\bm{\rho}_r).
$
Therefore, since
\(
D(\hat{\bm{\pi}}_r,\cdot)
\)
is closed, proper, and convex by assumption, the function
\[
f_r^k(\bm{\rho}_r)
=
D\!\left(\hat{\bm{\pi}}_r,\mathbf{C}_r^k\bm{\pi}_r(\bm{\rho}_r)\right)
+\frac{\gamma_\rho}{2}\|\bm{\rho}_r-\bm{\rho}_r^k\|_F^2
\]
is closed, proper, and convex on \(\mathcal D_r\). Since \(\gamma_\rho>0\), the term
\(
\frac{\gamma_\rho}{2}\|\bm{\rho}_r-\bm{\rho}_r^k\|_F^2
\)
is strongly convex, and hence
\begin{equation}
\label{eq:appendix_strong_convex_sum}
\sum_{r=1}^R f_r^k(\bm{\rho}_r)
\end{equation}
is strongly convex in the block variable \(\{\bm{\rho}_r\}_{r=1}^R\).

Next, \(\mathcal D_r\) is closed and convex for every \(r\), and each overlap map
$
\bm{\rho}_r \mapsto \bm{\rho}_r[r']
$
is linear. Hence the feasible set of \eqref{eq:proximal_state_subproblem_overlap} is closed and convex.

To verify feasibility, choose
\[
\bm{\rho}_r=\frac{1}{Q_r}\mathbf I_r,
\qquad r=1,\dots,R,
\]
and
\[
\bm{\rho}_{rr'}=\frac{1}{Q_{rr'}}\mathbf I_{rr'},
\qquad (r,r')\in\mathcal O,
\]
where \(Q_r=2^{q_r}\) and \(Q_{rr'}=2^{q_{rr'}}\). Then
\(
\bm{\rho}_r\in\mathcal D_r
\)
for all \(r\), and
\[
\bm{\rho}_r[r']
=
\frac{1}{Q_{rr'}}\mathbf I_{rr'}
=
\bm{\rho}_{rr'},
\qquad
\bm{\rho}_{r'}[r]
=
\frac{1}{Q_{rr'}}\mathbf I_{rr'}
=
\bm{\rho}_{r'r},
\]
for all \((r,r')\in\mathcal O\). Thus \eqref{eq:proximal_state_subproblem_overlap} is a feasible convex two-block problem in the variables
\[
\{\bm{\rho}_r\}_{r=1}^R
\qquad\text{and}\qquad
\{\bm{\rho}_{rr'}\}_{(r,r')\in\mathcal O}.
\]

The iterations \eqref{eq:admm_generic_updates} are precisely the two-block ADMM iterations for this problem. Since \eqref{eq:proximal_state_subproblem_overlap} is closed, proper, convex, feasible, and has an optimal solution by assumption, the standard two-block ADMM convergence theorem applies; see, e.g., \cite[Sec.~3.2.1 and Appendix~A]{Boyd2011} and \cite{EcksteinBertsekas1992}. Therefore, for any \(\beta>0\) and any initialization,
\begin{equation}
\label{eq:appendix_admm_residual_zero}
\bm{\delta}_{rr'}^\ell \to 0,
\qquad
\bm{\delta}_{r'r}^\ell \to 0,
\qquad
\forall (r,r')\in\mathcal O,
\end{equation}
and the objective values of \eqref{eq:proximal_state_subproblem_overlap} converge to its optimal value.

It remains to prove convergence of the primal sequence. First, the primal optimizer is unique. Let
\[
(\{\bm{\rho}_r\},\{\bm{\rho}_{rr'}\})
\quad\text{and}\quad
(\{\widetilde{\bm{\rho}}_r\},\{\widetilde{\bm{\rho}}_{rr'}\})
\]
be two primal optimal solutions of \eqref{eq:proximal_state_subproblem_overlap}. For any \(t\in(0,1)\), by convexity of the feasible set,
\[
(\{t\bm{\rho}_r+(1-t)\widetilde{\bm{\rho}}_r\},
 \{t\bm{\rho}_{rr'}+(1-t)\widetilde{\bm{\rho}}_{rr'}\})
\]
is feasible. Since \eqref{eq:appendix_strong_convex_sum} is strongly convex in \(\{\bm{\rho}_r\}\), if there exists \(r\) such that
\(
\bm{\rho}_r\neq \widetilde{\bm{\rho}}_r
\),
then
\begin{align*}
&\sum_{r=1}^R f_r^k\!\bigl(t\bm{\rho}_r+(1-t)\widetilde{\bm{\rho}}_r\bigr)
\\
&\qquad<
t\sum_{r=1}^R f_r^k(\bm{\rho}_r)
+
(1-t)\sum_{r=1}^R f_r^k(\widetilde{\bm{\rho}}_r),
\end{align*}
which contradicts optimality of both points. Hence
\begin{equation}
\label{eq:appendix_unique_rho_block}
\bm{\rho}_r=\widetilde{\bm{\rho}}_r,
\qquad
\forall r.
\end{equation}
Using the constraints in \eqref{eq:proximal_state_subproblem_overlap},
\[
\bm{\rho}_{rr'}
=
\bm{\rho}_r[r']
=
\widetilde{\bm{\rho}}_r[r']
=
\widetilde{\bm{\rho}}_{rr'},
\qquad
(r,r')\in\mathcal O,
\]
and therefore the primal optimizer of \eqref{eq:proximal_state_subproblem_overlap} is unique. Denote it by
\[
(\{\bm{\rho}_r^\star\}_{r=1}^R,\{\bm{\rho}_{rr'}^\star\}_{(r,r')\in\mathcal O}).
\]

Next, for each \(r\), the sequence \(\{\bm{\rho}_r^\ell\}_{\ell\ge 0}\) lies in \(\mathcal D_r\). Since
\[
\mathcal D_r
=
\{\bm{\rho}_r\succeq 0:\operatorname{Tr}(\bm{\rho}_r)=1\},
\]
one has
\[
\|\bm{\rho}_r^\ell\|_F
\le
\operatorname{Tr}(\bm{\rho}_r^\ell)
=
1,
\qquad \forall \ell,
\]
so \(\{\bm{\rho}_r^\ell\}_{\ell\ge 0}\) is bounded. Moreover,
\[
\bm{\rho}_{rr'}^\ell
=
\bm{\rho}_r^\ell[r']-\bm{\delta}_{rr'}^\ell,
\qquad
(r,r')\in\mathcal O,
\]
and each partial-trace map is linear and bounded, so \(\{\bm{\rho}_r^\ell[r']\}_{\ell\ge 0}\) is bounded. Together with \eqref{eq:appendix_admm_residual_zero}, this yields boundedness of \(\{\bm{\rho}_{rr'}^\ell\}_{\ell\ge 0}\). Hence the primal sequence
\begin{equation}
\label{eq:appendix_primal_sequence}
\left\{
\Bigl(\{\bm{\rho}_r^\ell\}_{r=1}^R,\{\bm{\rho}_{rr'}^\ell\}_{(r,r')\in\mathcal O}\Bigr)
\right\}_{\ell\ge 0}
\end{equation}
is bounded in a finite-dimensional space, and therefore has at least one cluster point. Let
\[
\Bigl(\{\bar{\bm{\rho}}_r\}_{r=1}^R,\{\bar{\bm{\rho}}_{rr'}\}_{(r,r')\in\mathcal O}\Bigr)
\]
be any cluster point of the sequence in \eqref{eq:appendix_primal_sequence}. By \cite[Sec.~3.2.1 and Appendix~A]{Boyd2011} and \cite[Thm.~8]{EcksteinBertsekas1992}, every cluster point of the primal sequence is a primal optimal solution. Since the primal optimal solution is unique, necessarily
\[
\bar{\bm{\rho}}_r=\bm{\rho}_r^\star,
\qquad r=1,\dots,R,
\]
and
\[
\bar{\bm{\rho}}_{rr'}=\bm{\rho}_{rr'}^\star,
\qquad (r,r')\in\mathcal O.
\]
Thus every cluster point of \eqref{eq:appendix_primal_sequence} coincides with
\[
\Bigl(\{\bm{\rho}_r^\star\}_{r=1}^R,\{\bm{\rho}_{rr'}^\star\}_{(r,r')\in\mathcal O}\Bigr).
\]
Therefore,
\[
\bm{\rho}_r^\ell \to \bm{\rho}_r^\star,
\qquad
r=1,\dots,R,
\]
and
\[
\bm{\rho}_{rr'}^\ell \to \bm{\rho}_{rr'}^\star,
\qquad
(r,r')\in\mathcal O.
\]

Hence the residuals satisfy \eqref{eq:appendix_admm_residual_zero}, the objective values converge to the optimal value of \eqref{eq:proximal_state_subproblem_overlap}, and the primal iterates converge to the unique solution.
\end{proof}

\section{Proof of Proposition~\ref{prop:quantitative_scaling_advantage}}
\label{app:proof_quantitative_scaling_advantage}

\begin{proof}
By definition,
\[
\begin{aligned}
P_{\mathrm{reg}}
&=
\sum_{r=1}^R \bigl[(4^{q_r}-1)+M_r(M_r-1)\bigr], \\
P_{\mathrm{glob}}
&=
(4^q-1)+M(M-1).
\end{aligned}
\]
Using \(4^{q_r}\le M_r\le \mu\,4^{q_r}\) for all \(r\),
\begin{align*}
P_{\mathrm{reg}}
&\le \sum_{r=1}^R \bigl(4^{q_r}+M_r^2\bigr)
 \le \sum_{r=1}^R \bigl(4^{q_r}+\mu^2 16^{q_r}\bigr) \\
&\le (1+\mu^2)\sum_{r=1}^R 16^{q_r}
 \le (1+\mu^2)\,R\,16^{q_{\max}},
\end{align*}
while, since \(M\ge 4^q\),
\[
P_{\mathrm{glob}}
\ge M(M-1)\ge 4^q(4^q-1)=16^q-4^q.
\]
Hence
\[
F_{\mathrm{mem}}
=
\frac{P_{\mathrm{glob}}}{P_{\mathrm{reg}}}
\ge
\frac{16^q-4^q}{(1+\mu^2)\,R\,16^{q_{\max}}}.
\]

Next, by definition,
\[
N_{\mathrm{comm}}
=
\sum_{r=1}^R\sum_{r':(r,r')\in\mathcal O}4^{q_{rr'}}.
\]
Since \(q_{rr'}\le q_{\mathrm{ov},\max}\) for all \((r,r')\in\mathcal O\), and each \(r\) has degree at most \(d_{\max}\),
\[
N_{\mathrm{comm}}
\le
\sum_{r=1}^R\sum_{r':(r,r')\in\mathcal O}4^{q_{\mathrm{ov},\max}}
\le
d_{\max}R\,4^{q_{\mathrm{ov},\max}}.
\]

Also, using \(M_r\ge 4^{q_r}\),
\begin{align*}
P_{\mathrm{reg}}
&\ge \sum_{r=1}^R M_r(M_r-1)
 \ge \sum_{r=1}^R 4^{q_r}(4^{q_r}-1) \\
&= \sum_{r=1}^R \bigl(16^{q_r}-4^{q_r}\bigr)
 \ge R\bigl(16^{q_{\min}}-4^{q_{\min}}\bigr).
\end{align*}
Therefore,
\[
\frac{N_{\mathrm{comm}}}{P_{\mathrm{reg}}}
\le
\frac{d_{\max}R\,4^{q_{\mathrm{ov},\max}}}{R(16^{q_{\min}}-4^{q_{\min}})}
=
\frac{d_{\max}}{1-4^{-q_{\min}}}\,
4^{\,q_{\mathrm{ov},\max}-2q_{\min}}.
\]

If \(q_{\max}\), \(q_{\mathrm{ov},\max}\), and \(d_{\max}\) remain uniformly bounded as \(q\) grows, then
\[
P_{\mathrm{reg}}=O(R),\qquad
N_{\mathrm{comm}}=O(R),\qquad
P_{\mathrm{glob}}=O(16^q),
\]
and the lower bound
\[
F_{\mathrm{mem}}
\ge
\frac{16^q-4^q}{(1+\mu^2)\,R\,16^{q_{\max}}}
\]
grows exponentially in \(q-q_{\max}\).
\end{proof}

\section{Proof of Remark~\ref{rem:kl_ml_map}}
\label{app:proof_kl_ml_map}

\begin{proof}
Fix a region \(r\). Under the noisy model, the recorded outcome at shot \(t\) is categorical with distribution \(\mathbf C_r\bm{\pi}_r(\bm{\rho}_r)\), i.e.,
\[
\Pr(y_{rt}=m\mid \bm{\rho}_r,\mathbf C_r)
=
[\mathbf C_r\bm{\pi}_r(\bm{\rho}_r)]_m,
\qquad m=1,\dots,M_r .
\]
Using the one-hot representation \(\mathbf y_{rt}\), this gives
\[
\Pr(\mathbf y_{rt}\mid \bm{\rho}_r,\mathbf C_r)
=
\prod_{m=1}^{M_r}
\bigl([\mathbf C_r\bm{\pi}_r(\bm{\rho}_r)]_m\bigr)^{[\mathbf y_{rt}]_m}.
\]
Hence, for \(T_r\) conditionally i.i.d.\ shots,
\[
\Pr(\{\mathbf y_{rt}\}_{t=1}^{T_r}\mid \bm{\rho}_r,\mathbf C_r)
=
\prod_{t=1}^{T_r}\prod_{m=1}^{M_r}
\bigl([\mathbf C_r\bm{\pi}_r(\bm{\rho}_r)]_m\bigr)^{[\mathbf y_{rt}]_m},
\]
and therefore
\begin{equation*}
\begin{aligned}
\log \Pr(\{\mathbf y_{rt}\}_{t=1}^{T_r} \mid  & \bm{\rho}_r,\mathbf C_r)
=
\\
& \sum_{t=1}^{T_r}\sum_{m=1}^{M_r}
[\mathbf y_{rt}]_m
\log\bigl([\mathbf C_r\bm{\pi}_r(\bm{\rho}_r)]_m\bigr).
\end{aligned}
\end{equation*}
By \eqref{eq:empirical_frequencies},
\[
\hat{\pi}_{r,m}
=
\frac{1}{T_r}\sum_{t=1}^{T_r}[\mathbf y_{rt}]_m,
\qquad m=1,\dots,M_r,
\]
so
\[
\log \Pr(\{\mathbf y_{rt}\}_{t=1}^{T_r}\mid \bm{\rho}_r,\mathbf C_r)
=
T_r\sum_{m=1}^{M_r}
\hat{\pi}_{r,m}\log\bigl([\mathbf C_r\bm{\pi}_r(\bm{\rho}_r)]_m\bigr).
\]
Thus
\begin{equation*}
\begin{aligned}
-\log \Pr(\{\mathbf y_{rt}\}_{t=1}^{T_r} \mid & \bm{\rho}_r,\mathbf C_r)
=
\\
& -T_r\sum_{m=1}^{M_r}\hat{\pi}_{r,m}\,
\log\bigl([\mathbf C_r\bm{\pi}_r(\bm{\rho}_r)]_m\bigr).
\end{aligned}
\end{equation*}
Adding \(T_r\sum_{m=1}^{M_r}\hat{\pi}_{r,m}\log \hat{\pi}_{r,m}\) (a constant) to both sides 
\[
-\log \Pr(\{\mathbf y_{rt}\}_{t=1}^{T_r}\mid \bm{\rho}_r,\mathbf C_r)
+
T_r\sum_{m=1}^{M_r}\hat{\pi}_{r,m}\log \hat{\pi}_{r,m}
\]
\[
=
T_r\sum_{m=1}^{M_r}\hat{\pi}_{r,m}
\log\frac{\hat{\pi}_{r,m}}{[\mathbf C_r\bm{\pi}_r(\bm{\rho}_r)]_m}
=
T_r\,\mathrm{KL} \bigl(\hat{\bm{\pi}}_r\,\|\,\mathbf C_r\bm{\pi}_r(\bm{\rho}_r)\bigr).
\]
Hence, up to an additive term depending only on the data,
\begin{align}
\nonumber
-\log \Pr(\{\mathbf y_{rt}\}_{t=1}^{T_r}\mid \bm{\rho}_r,\mathbf C_r)
&=
T_r\,\mathrm{KL} \bigl(\hat{\bm{\pi}}_r\,\|\,\mathbf C_r\bm{\pi}_r(\bm{\rho}_r)\bigr)
\\
& \qquad + \mathrm{const}.
\nonumber 
\end{align}
Summing over \(r=1,\dots,R\), the joint negative log-likelihood is, up to an additive constant independent of \((\{\bm{\rho}_r\},\{\mathbf C_r\})\),
\[
\sum_{r=1}^R
T_r\,\mathrm{KL} \bigl(\hat{\bm{\pi}}_r\,\|\,\mathbf C_r\bm{\pi}_r(\bm{\rho}_r)\bigr).
\]
Therefore, under \eqref{eq:multi_region_estimation_problem_state}--\eqref{eq:multi_region_estimation_problem_overlap}, the exact likelihood-equivalent form is the shot-weighted KL objective above. If all regions use the same number of shots, the factors \(T_r\) are equal and can be dropped without changing the minimizers, so \eqref{eq:multi_region_estimation_problem} is a constrained maximum-likelihood estimator (MLE). Adding a regularizer using the log-prior of \((\{\bm{\rho}_r\},\{\mathbf C_r\})\) gives the posterior distribution 
\begin{align}
\nonumber
& -\log \Pr( \{\bm{\rho}_r\},\{\mathbf C_r\} \mid \{\mathbf y_{rt}\}_{r,t}) = 
\\
\nonumber 
&\qquad -\log \Pr(\{\mathbf y_{rt}\}_{r,t}\mid \{\bm{\rho}_r\},\{\mathbf C_r\})
-\log \Pr \bigl(\{\bm{\rho}_r\},\{\mathbf C_r\}\bigr).
\end{align}
Minimizing this under \eqref{eq:multi_region_estimation_problem_state}--\eqref{eq:multi_region_estimation_problem_overlap}, yields  a constrained maximum a posteriori (MAP) estimator.
\end{proof}

\bibliographystyle{IEEEtran}
\bibliography{references}

\end{document}